\documentclass[twoside]{article}
\usepackage{amsfonts, amsthm}
\usepackage[mathcal]{eucal}
\usepackage[dvips]{graphicx}

\newtheorem{theorem}{Theorem}[section]
\newtheorem{lemma}[theorem]{Lemma}
\newtheorem{prop}[theorem]{Proposition}
\newtheorem{cor}[theorem]{Corollary}
\newtheorem{remark}[theorem]{Remark}
\newtheorem{definition}[theorem]{Definition}

\def\l{\lambda}
\def\a{\alpha}
\def\b{\beta}

\def\R{\mathbb{R}}

\def\C{\mathbb{C}}
\def\N{\mathbb{N}}
\def\Z{\mathbb{Z}}
\def\P{\mathbb{P}}

\def\T{\mathbb{T}}

\def\M{\mathcal{M}}
\def\Mba{\M_{\beta, \alpha}}
\def\Mmba{\M_{\beta, \alpha}}

\def\W{\mathcal{W}}
\def\P{\mathcal{P}}
\def\Pba{\P_{\beta, \alpha}}

\def\n{\nabla}
\def\nba{\nabla_{b,a}}
\def\cds{\!\cdot_\textbf{s}\!}
\def\mbf2{\mathbf{2}}

\begin{document}

\title{Global Birkhoff coordinates \\ for the periodic Toda lattice}
\author{Andreas Henrici \and Thomas Kappeler\footnote{Supported in part by the Swiss National Science Foundation, the programme SPECT, and the European Community through the FP6 Marie Curie RTN ENIGMA (MRTN-CT-2004-5652)}}

\maketitle

\begin{abstract}
In this paper we prove that the periodic Toda lattice admits globally defined Birkhoff coordinates.
\end{abstract}

%\tableofcontents

\section{Introduction} \label{introduction}

Consider the Toda lattice with period $N$ ($N \geq 2$), 
\begin{displaymath}
\dot{q}_n = \partial_{p_n} H, \quad \dot{p}_n = -\partial_{q_n} H
\end{displaymath}
for $n \in \Z$, where the (real) coordinates $(q_n, p_n)_{n \in \Z}$ satisfy $(q_{n+N}, p_{n+N}) = (q_n, p_n)$ for any $n \in \Z$ and the Hamiltonian $H_{Toda}$ is given by
\begin{equation} \label{hamtodapq}
H_{Toda} = \frac{1}{2} \sum_{n=1}^N p_n^2 + \alpha^2 \sum_{n=1}^N e^{q_n - q_{n+1}}
\end{equation}
where $\a$ is a positive parameter, $\a > 0$. For the standard Toda lattice, $\a=1$. The Toda lattice was introduced by Toda \cite{toda} and studied extensively in the sequel. It is an important model for an integrable system of $N$ particles in one space dimension with nearest neighbor interaction and belongs to the family of lattices introduce and numerically investigated by Fermi, Pasta, and Ulam in their seminal paper \cite{fpu}. To prove the integrability of the Toda lattice, Flaschka introduced in \cite{fla1} the (noncanonical) coordinates
\begin{displaymath}
b_n := -p_n \in \R, \quad a_n := \alpha e^{\frac{1}{2} (q_n - q_{n+1})} \in \R_{>0} \quad (n \in \Z).
\end{displaymath}
These coordinates describe the motion of the Toda lattice relative to the center of mass. Note that the total momentum is conserved by the Toda flow, hence any trajectory of the center of mass is a straight line.

In these coordinates the Hamiltonian $H_{Toda}$ takes the simple form
\begin{displaymath}
H = \frac{1}{2} \sum_{n=1}^N b_n^2 + \sum_{n=1}^N a_n^2,
\end{displaymath}
and the equations of motion are
\begin{equation} \label{flaeqn}
\left\{ \begin{array}{lllll}
 \dot{b}_n & = & a_n^2 - a_{n-1}^2 \\
 \dot{a}_n & = & \frac{1}{2} a_n (b_{n+1} - b_n)
\end{array} \right. \qquad (n \in \Z).
\end{equation}
Note that $(b_{n+N}, a_{n+N}) = (b_n, a_n)$ for any $n \in \Z$, and $\prod_{n=1}^N a_n = \alpha^N$.
%Whereas in the physical origin and the corresponding coordinates $(q, p)$ the value of the parameter $\alpha > 0$ is fixed, the Flaschka coordinates $(b, a)$ have the advantage that it is possible to treat all values of $\alpha > 0$ simultaneously. Consequently
Hence we can identify the sequences $(b_n)_{n \in \Z}$ and $(a_n)_{n \in \Z}$ with the vectors $(b_n)_{1 \leq n \leq N} \in \R^N$ and $(a_n)_{1 \leq n \leq N} \in \R_{>0}^N$. Our aim is to study the normal form of the system of equations (\ref{flaeqn}) on the phase space
\begin{displaymath}
\M := \R^N \times \R_{>0}^N.
\end{displaymath}
This system is Hamiltonian with respect to the nonstandard Poisson structure $J \equiv J_{b,a}$, defined at a point $(b,a) = ((b_n, a_n)_{1 \leq n \leq N}$ by
\begin{equation} \label{jdef}
J = \left( \begin{array}{cc}
0 & A \\
-{}^t A & 0 \\
\end{array} \right),
\end{equation}
where $A$ is the $b$-independent $N \times N$-matrix
\begin{equation} \label{adef}
A = \frac{1}{2} \left( \begin{array}{ccccc}
a_1 & 0 & \ldots & 0 & -a_N \\
-a_1 & a_2 & 0 & \ddots & 0 \\
0 & -a_2 & a_3 & \ddots & \vdots \\
\vdots & \ddots & \ddots & \ddots & 0 \\
0 & \ldots & 0 & -a_{N-1} & a_N \\
\end{array} \right).
\end{equation}
The Poisson bracket corresponding to (\ref{jdef}) is then given by
\begin{eqnarray}
\{ F, G \}_J(b,a) & = & \langle (\nabla_b F, \nabla_a F), \, J \, (\nabla_b G, \nabla_a G) \rangle_{\R^{2N}} \nonumber\\
& = & \langle \nabla_b F, A \, \nabla_a G \rangle_{\R^N} - \langle \nabla_a F, A^t \, \nabla_b G \rangle_{\R^N}. \label{poisson}
\end{eqnarray}
%\begin{equation} \label{poisson}
%\{ F, G \}(b,a) = \langle \textrm{grad}_{(b,a)}F, J_{(b,a)} \textrm{grad}_{(b,a)}G \rangle,
%\end{equation}
where $F,G \in C^1(\M)$ and where $\nabla_b$ and $\nabla_a$ denote the gradients with respect to the $N$-vectors $b = (b_1, \ldots, b_N)$ and $a = (a_1, \ldots, a_N)$, respectively. Therefore, equations (\ref{flaeqn}) can alternatively be written as $\dot{b}_n = \{ b_n, H \}_J$, $\dot{a}_n = \{ a_n, H \}_J$ $(1 \leq n \leq N)$. Further note that
\begin{equation}
  \{ b_n, a_n \}_J = \frac{a_n}{2}; \quad \{ b_{n+1}, a_n \}_J = -\frac{a_n}{2},
\end{equation}
while $\{ b_n, a_k \}_J = 0$ for any $n,k$ with $n \notin \{ k, k+1 \}$.
%\begin{displaymath} %\label{todapoisson}
%\left\{ \begin{array}{ccc}
% \dot{b}_n & = & \{ b_n, H \} \\
% \dot{a}_n & = & \{ a_n, H \}
%\end{array} \right..
%\end{displaymath}

Since the matrix $A$ defined by (\ref{adef}) has rank $N-1$, the Poisson structure $J$ is degenerate.
%which reflects the fact that the representation of the periodic Toda lattice in the Flaschka coordinates $(b,a)$ ``artificially'' unifies the treatment of all possible parameter values for $\alpha$ in one system, whereas the system physically only depends on the $N-1$ relative coordinates $q_2 - q_1, \ldots, q_N - q_{N-1}$. Mathematically, this degeneracy becomes clear through the fact that the Poisson structure $J$ 
It admits the two Casimir functions\footnote{A smooth function $C: \M \to \R$ is a Casimir function for $J$ if $\{ C, \cdot \}_J \equiv 0$.}
\begin{equation} \label{casimirdef}
C_1 := -\frac{1}{N} \sum_{n=1}^N b_n \quad \textrm{and} \quad C _2 := \left( \prod_{n=1}^N a_n \right)^\frac{1}{N}
\end{equation}
whose gradients $\nba C_i \! = \! (\n_b C_i, \n_a C_i)$ ($i = 1,2$), given by
\setlength\arraycolsep{1.5pt} {
\begin{eqnarray}
\nabla_b C_1 & = & -\frac{1}{N} (1, \ldots, 1), \qquad \nabla_a C_1 = 0, \label{c1grad} \\
\nabla_b C_2 & = & 0, \qquad \nabla_a C_2 = \frac{C_2}{N} \left( \frac{1}{a_1}, \ldots, \frac{1}{a_N} \right), \label{c2grad}
\end{eqnarray}}
are linearly independent at each point $(b,a)$ of $\M$.

Let
\begin{displaymath}
\Mba := \{ (b,a) \in \R^{2N} : (C_1, C_2) = (\beta, \a) \}
\end{displaymath}
denote the level set of $(C_1, C_2)$ for $(\b, \a) \in \R \times \R_{> 0}$. Note that $(-\b 1_N, \a 1_N) \in \Mba$ where $1_N = (1, \ldots, 1) \in \R^N$. By (\ref{c1grad})-(\ref{c2grad}), the sets $\Mba$ are real analytic submanifolds of $\M$ of codimension two. Furthermore the Poisson structure $J$, restricted to $\Mba$, becomes nondegenerate everywhere on $\Mba$ and therefore induces a symplectic structure $\nu_{\b, \a}$ on $\Mba$. In this way, we obtain a symplectic foliation of $\M$ with $\Mba$ being its (symplectic) leaves.

To state the main result of this paper, we introduce the model space
$$ \P := \R^{2(N-1)} \times \R \times \R_{>0} $$
 endowed with the degenerate Poisson structure $J_0$ whose symplectic leaves are $\R^{2(N-1)} \times \{ \b \} \times \{ \a \}$ endowed with the canonical symplectic structure.
\begin{theorem} \label{sumthm}
There exists a map
\begin{displaymath}
\begin{array}{ccll}
 \Omega: & (\M, J) & \to & (\P, J_0) \\
 & (b,a) & \mapsto & ((x_n, y_n)_{1 \leq n \leq N-1}, C_1, C_2)
\end{array}
\end{displaymath}
with the following properties:
\begin{itemize}
  \item $\Omega$ is a real analytic diffeomorphism.
  \item $\Omega$ is canonical, i.e. it preserves the  Poisson brackets. In particular, the symplectic foliation of $\M$ by $\Mba$ is trivial.
  \item The coordinates $(x_n, y_n)_{1 \leq n \leq N-1}, C_1, C_2$ are
  global Birkhoff coordinates for the periodic Toda lattice, i.e. the transformed Toda Hamiltonian $\hat{H} = H \circ \Omega^{-1}$
  is a function of the actions $(I_n)_{1 \leq n \leq N-1}$ and $C_1, C_2$ alone.
\end{itemize}
\end{theorem}
Further properties of $\Omega$ are discussed at the end of section \ref{diffeochapter}.

In \cite{ahtk3}, we use the Birkhoff coordinates $(x_n, y_n)_{1 \leq n \leq N-1}$ given by Theorem \ref{sumthm} to obtain a KAM-theorem for the periodic Toda lattice. 

\emph{Related work:} Theorem \ref{sumthm} improves on earlier work on the normal form of the periodic Toda lattice in \cite{bbgk, bggk}. In particular, we construct global Birkhoff coordinates on all of $\M$ instead of a single symplectic leaf and show that techniques recently developed for treating the KdV equation (cf. \cite{kama, kapo}) and the defocusing NLS equation (cf. \cite{gkp, mcva1}) can also be applied for the Toda lattice.

\emph{Outline of the paper:} In section \ref{tools} we review the Lax pair of the periodic Toda lattice and collect some auxiliary results on the spectrum of the Jacobi matrix $L(b,a)$ associated to an element $(b,a) \in \M$. The construction of the coordinates $(x_n, y_n)_{1 \leq n \leq N-1}$ (see section \ref{birkhoff}) uses the action variables $(I_n)_{1 \leq n \leq N-1}$ defined on $\M$ and the angle variables $(\theta_n)_{1 \leq n \leq N-1}$ defined on a dense domain $\M' \subset \M$, both of which have been studied in detail in our previous paper \cite{ahtk1}. We give a brief review of these action-angle variables in section \ref{review}. The coordinates $(x_n, y_n)_{1 \leq n \leq N-1}$ are then defined on the dense domain $\M'$ by
$$ (x_n, y_n) = \sqrt{2 I_n} (\cos \theta_n, \sin \theta_n). $$
In a first step we show that the coordinate functions $(x_n, y_n)$ extend to real analytic functions on all of $\M$. Using the canonical relations among the action-angle variables established in \cite{ahtk1} we then show in section \ref{diffeochapter} that $\Omega$ is a canonical local diffeomorphism. Finally, the fact that $\Omega$ is 1-1 and onto is deduced from a priori estimates of the actions which are used to prove that $\Omega$ is proper (cf. section \ref{diffeochapter}).

\section{Preliminaries} \label{tools}

It is well known (cf. e.g. \cite{toda}) that the system (\ref{flaeqn}) admits a Lax pair formulation $\dot{L}= \frac{\partial L}{\partial t} = [B, L]$, where $L \equiv L^+(b,a)$ is the periodic Jacobi matrix defined by
\begin{equation} \label{jacobi}
L^\pm(b,a) := \left( \begin{array}{ccccc}
b_1 & a_1 & 0 & \ldots & \pm a_N \\
a_1 & b_2 & a_2 & \ddots & \vdots \\
0 & a_2 & b_3 & \ddots & 0 \\
\vdots & \ddots & \ddots & \ddots & a_{N-1} \\
\pm a_N & \ldots & 0 & a_{N-1} & b_N \\
\end{array} \right),
\end{equation}
and $B$ the skew-symmetric matrix
\begin{displaymath}
B = \left( \begin{array}{ccccc}
0 & a_1 & 0 & \ldots & -a_N \\
-a_1 & 0 & a_2 & \ddots & \vdots \\
0 & -a_2 & \ddots & \ddots & 0 \\
\vdots & \ddots & \ddots & \ddots & a_{N-1} \\
a_N & \ldots & 0 & -a_{N-1} & 0 \\
\end{array} \right).
\end{displaymath}
Hence the flow of $\dot{L} = [B, L]$ is isospectral.
\begin{prop} \label{isosp}
For a solution $\big( b(t), a(t) \big)$ of the periodic Toda lattice (\ref{flaeqn}), the eigenvalues $(\l_j^\pm)_{1 \leq j \leq N}$ of $L^\pm\big( b(t), a(t) \big)$ are conserved quantities.
\end{prop}
%In addition, Manakov \cite{mana} showed that these eigenvalues are in involution with each other, i.e. $\{ \l_i, \l_j \} = 0$ for all $ 1 \leq i,j \leq N$.

Let us now collect a few results from \cite{moer} and \cite{toda} of the spectral theory of Jacobi matrices needed in the sequel. Denote by $\M^\C$ the complexification of the phase space $\M$,
\begin{displaymath}
\M^{\C} = \{ (b,a) \in \C^{2N} : \textrm{Re }a_j > 0 \quad \forall \, 1 \leq j \leq N \}.
\end{displaymath}
For $(b,a) \in \M^\C$ we consider for any complex number $\l$ the difference equation
\begin{equation} \label{diff}
(R_{b,a} y)(k) = \l y(k) \quad (k \in \Z)
\end{equation}
where $y(\cdot) = y(k)_{k \in \Z} \in \C^{\Z}$ and $R_{b,a}$ is the difference operator
\begin{equation} \label{rdef}
R_{b,a} = a_{k-1} S^{-1} + b_k S^0 + a_k S^1
\end{equation}
with $S^m$ denoting the shift operator of order $m \in \Z$, i.e.
\begin{displaymath}
(S^m y)(k) = y(k+m) \textrm{ for } k \in \Z.
\end{displaymath}

\emph{Fundamental solutions:} The two fundamental solutions $y_1(\cdot, \l)$ and $y_2(\cdot, \l)$ of (\ref{diff}) are defined by the standard initial conditions $y_1(0, \l) = 1$, $y_1(1, \l) = 0$ and $y_2(0, \l) = 0$, $y_2(1, \l) = 1$. They satisfy the \emph{Wronskian identity}
\begin{equation} \label{wronski}
W(n) := y_1(n, \l) y_2(n+1, \l) - y_1(n+1, \l) y_2(n, \l) = \frac{a_N}{a_n}.
\end{equation}
Note that for $n = N$ one gets
\begin{equation} \label{wronskispecial}
W(N) = 1.
\end{equation}
For each $k \in \N$, $y_i(k, \l, b, a)$, $i = 1,2$, is a polynomial in $\l$ of degree at most $k-1$ and depends real analytically on $(b,a)$ (see \cite{moer}). In particular, one easily verifies that $y_2(N+1, \l, b, a)$ is a polynomial in $\l$ of degree $N$ with leading term $\a^{-N} \l^N$, whereas $y_1(N, \l)$ is a polynomial in $\l$ of degree less than $N$.

\emph{Discriminant:} We denote by $\Delta(\l) \equiv \Delta(\l, b, a)$ the \emph{discriminant} of (\ref{diff}), defined by
\begin{equation} \label{discrdef}
\Delta(\l) := y_1(N, \l) + y_2(N+1, \l).
\end{equation}
In the sequel, we will often write $\Delta_\l$ for $\Delta(\l)$. %Note that $y_2(N+1, \l)$ is a polynomial in $\l$ of degree $N$ with leading term $\a^{-N} \l^N$, 
As $y_2(N+1,\l) = \a^{-N} \l^N + \ldots$ and $y_1(N,\l) = O(\l^{N-1})$, $\Delta(\l, b, a)$ is a polynomial in $\l$ of degree $N$ with leading term $\a^{-N} \l^N$, and it depends real analytically on $(b, a)$ (see e.g. \cite{toda}). According to Floquet's Theorem (see e.g. \cite{teschl2}), for $\l \in \C$ given, (\ref{diff}) admits a periodic or antiperiodic solution of period $N$ if the discriminant $\Delta(\l)$ satisfies $\Delta(\l) = +2$ or $\Delta(\l) = -2$, respectively. These solutions correspond to eigenvectors of $L^+$ or $L^-$, respectively, with $L^\pm$ defined by (\ref{jacobi}). It turns out to be more convenient to combine these two cases by considering the periodic Jacobi matrix $Q \equiv Q(b,a)$ of size $2N$ defined by
\begin{displaymath}
Q = \left( \begin{array}{cccc|cccc}
b_1 & a_1 & \ldots & 0 & 0 & \ldots & 0 & a_N \\
a_1 & b_2 & \ddots & \vdots & 0 & \ldots & & 0 \\
\vdots & \ddots & \ddots & a_{N-1} & \vdots & & & \vdots \\
0 & \ddots & \;\; a_{N-1} & b_N & a_N & \ldots & 0 & 0 \\
\hline
0 & \ldots & 0 & a_N & b_1 & a_1 & \ldots & 0 \\
0 & \ldots & & 0 & a_1 & b_2 & \ddots & \vdots \\
\vdots & & & \vdots & \vdots & \ddots & \ddots & a_{N-1} \\
a_N & \ldots & 0 & 0 & 0 & \ddots & \;\; a_{N-1} & b_N \\
\end{array} \right).
\end{displaymath}
Then the spectrum of the matrix $Q$ is the union of the spectra of the matrices $L^+$ and $L^-$ and therefore the zero set of the polynomial $\Delta^2_\l - 4$.  The function $\Delta^2_\l - 4$ is a polynomial in $\l$ of degree $2N$ and admits a product representation
\begin{equation} \label{delta2lrepr}
  \Delta^2_\l - 4 = \alpha^{-2N} \prod_{j=1}^{2N} (\l - \l_j).
\end{equation}
The factor $\a^{-2N}$ in (\ref{delta2lrepr}) comes from the above mentioned fact that the leading term of $\Delta(\l)$ is $\a^{-N} \l^N$.

For any $(b,a) \in \M$, the matrix $Q$ is symmetric and hence the eigenvalues $(\l_j)_{1 \leq j \leq 2N}$ of $Q$ are real. When listed in increasing order and with their algebraic multiplicities, they satisfy the following relations (cf. \cite{moer})
\begin{displaymath}
\l_1 < \l_2 \leq \l_3 < \l_4 \leq \l_5 < \ldots \l_{2N-2} \leq \l_{2N-1} < \l_{2N}.
\end{displaymath}
As explained above, the $\l_j$ are periodic or antiperiodic eigenvalues of $L$ and thus eigenvalues of $L^+$ or $L^-$ according to whether $\Delta(\l_j) = 2$ or $\Delta(\l_j) = -2$. One has (cf. \cite{moer})
\begin{equation} \label{deltalambdapm2}
  \Delta(\l_1) = (-1)^N \cdot 2, \quad \Delta(\l_{2n}) = \Delta(\l_{2n+1}) = (-1)^{n+N} \cdot 2, \quad \Delta(\l_{2N}) = 2.
\end{equation}

Since $\Delta_\l$ is a polynomial of degree $N$ with leading term $\a^{-N} \l^N$, $\dot{\Delta}_\l \equiv \dot{\Delta}(\l) = \frac{d}{d\l} \Delta(\l)$ is a polynomial of degree $N-1$ with leading term $N \a^{-N} \l^{N-1}$, hence admits a product representation of the form
\begin{equation} \label{dotlrepr}
  \dot{\Delta}_\l = N \a^{-N} \prod_{k=1}^{N-1} (\l - \dot{\l}_k).
\end{equation}
The zeroes $(\dot{\l}_n)_{1 \leq n \leq N-1}$ of $\dot{\Delta}_\l$ satisfy $\l_{2n} \leq \dot{\l}_n \leq \l_{2n+1}$ for any $1 \leq n \leq N-1$. The open intervals $(\l_{2n}, \l_{2n+1})$ are referred to as the \emph{$n$-th spectral gap} and $\gamma_n := \l_{2n+1} - \l_{2n}$ as the \emph{$n$-th gap length}. Note that $|\Delta(\l)| > 2$ on the spectral gaps. We say that the $n$-th gap is \emph{open} if $\gamma_n > 0$ and \emph{collapsed} otherwise. The set of elements $(b,a) \in \M$ for which the $n$-th gap is collapsed is denoted by $D_n$,
\begin{equation} \label{dndef}
D_n := \{ (b,a) \in \M : \gamma_n = 0 \}.
\end{equation}
Using that $\gamma_n^2$ (unlike $\gamma_n$) is a real analytic function on $\M$, it can be shown that $D_n$ is a real analytic submanifold of $\M$ of codimension $2$ (cf. \cite{kapo} for a similar statement in the case of Hill's operator).

\emph{Isolating neighborhoods:}
Let $(b,a) \in \M$ be given. The strict inequalities $\l_{2n-1} < \l_{2n}$ ($1 \leq n \leq N$) guarantee the existence of a family of mutually disjoint open subsets $(U_n)_{1 \leq n \leq N-1}$ of $\C$ so that for any $1 \leq n \leq N-1$, $U_n$ is a neighborhood of the closed interval $[\l_{2n}, \l_{2n+1}]$. Such a family of neighborhoods is referred to as a family of \emph{isolating neighborhoods} for $(b,a)$.

In the case where $(b,a) \in \M^\C$, we list the eigenvalues $(\l_j)_{1 \leq j \leq 2N}$ in lexicographic ordering\footnote{The lexicographic ordering $a \prec b$ for complex numbers $a$ and $b$ is defined by
\begin{equation} \label{lexorder}
a \prec b \quad :\Longleftrightarrow \quad \left\{
\begin{array}{l} \textrm{Re
}a < \textrm{Re }b \\
\textrm{or} \\
\textrm{Re }a = \textrm{Re }b \textrm{ and } \textrm{Im }a \leq
\textrm{Im }b. \end{array} \right.
\end{equation}
}
\begin{displaymath}
  \l_1 \prec \l_2 \prec \l_3 \prec \ldots \prec \l_{2N}.
\end{displaymath}
We then extend the gap lenghts $\gamma_n$ to all of $\M^\C$ by
\begin{displaymath}
  \gamma_n := \l_{2n+1} - \l_{2n} \quad (1 \leq n \leq N-1)
\end{displaymath}
and define
\begin{equation} \label{dncdef}
  D_n^\C := \{ (b,a) \in \M^{\C} : \gamma_n = 0 \}.
\end{equation}
In the sequel, we will omit the superscript and always write $D_n$ for $D_n^\C$.

Similarly, we do this for the zeroes $(\dot{\l}_n)_{1 \leq n \leq N-1}$ of $\dot{\Delta}_\l$. As the lexicographic ordering is not continuous, the $\l_i$'s and $\dot{\l}_i$'s no longer depend continuously on $(b,a) \in \M^\C$. However, if we choose a small enough complex neighborhood $\W$ of $\M$ in $\M^\C$, then for any $(b,a) \in \W$ the closed intervals $G_n \subseteq \C$ ($1 \leq n \leq N-1$) defined by
\begin{equation} \label{gndef}
G_n := \{ (1-t) \l_{2n} + t \l_{2n+1}: 0 \leq t \leq 1 \}
\end{equation}
 are pairwise disjoint, and hence, as in the real case, there exists a family of isolating neighborhoods $(U_n)_{1 \leq n \leq N-1}$.

\begin{lemma} \label{wmungnlemma}
There exists a neighborhood $\W$ of $\M$ in $\M^\C$ such that for any $(b,a) \in \W$, there are neighborhoods $U_n$ of $G_n$ in $\C$ ($1 \leq n \leq N-1$) which are pairwise disjoint.
\end{lemma}

\begin{remark}
In the sequel, we will have to shrink the complex neighborhood $\W$ several times, but continue to denote it by the same letter.
\end{remark}

\emph{Contours $\Gamma_n$:} For any $(b,a) \in \W$ and any $1 \leq n \leq N-1$, we denote by $\Gamma_n$ a circuit in $U_n$ around $G_n$ with counterclockwise orientation. 

\emph{Isospectral set:} For $(b,a) \in \M$, the set Iso$(b,a)$ of all elements $(b',a') \in \M$ so that $Q(b',a')$ has the same spectrum as $Q(b,a)$ is described with the help of the Dirichlet eigenvalues $\mu_1 < \mu_2 < \ldots < \mu_{N-1}$ of (\ref{diff}) defined by
\begin{equation} \label{mundef}
y_1(N+1, \mu_n) = 0.
\end{equation}
They coincide with the eigenvalues of the $(N-1) \times (N-1)$-matrix $L_2 = L_2(b,a)$ given by
$$ \left( \begin{array}{ccccc}
b_2 & a_2 & 0 & \ldots & 0  \\
a_2 & \ddots & \ddots & \ddots & \vdots  \\
0 & \ddots & \ddots & \ddots & 0 \\
\vdots & \ddots & \ddots & \ddots & a_{N-1} \\
0 & \ldots & 0 & a_{N-1} & b_N \\
\end{array} \right). $$
In the sequel, we will also refer to $\mu_1, \ldots, \mu_{N-1}$ as the Dirichlet eigenvalues of $L(b,a)$. Evaluating the Wronskian identity (\ref{wronski}) at $\l = \mu_n$ one sees that $\mu_n$ lies in the closure of the $n$-th spectral gap. More precisely, substituting $y_1(N+1, \mu_n) = 0$ in the identity (\ref{wronski}) with $\l = \mu_n$ yields
\begin{equation} \label{wrmu}
y_1(N,\mu_n) y_2(N+1,\mu_n) = 1.
\end{equation}
Hence the value of the discriminant at $\mu_n$ is given by
\begin{equation} \label{discdir}
\Delta(\mu_n) = y_2(N+1,\mu_n) + \frac{1}{y_2(N+1,\mu_n)}
\end{equation}
and $|\Delta(\mu_n)| \geq 2$. By Lemma \ref{speclemma} below, given the point $(b,a)$ with $b_1 = \ldots = b_N = \b$ and $a_1 = \ldots = a_N = \a$, one has $\l_{2n} = \l_{2n+1}$ and hence $\mu_n = \l_{2n}$ for any $1 \leq n \leq N-1$. It then follows from a straightforward deformation argument that $\l_{2n} \leq \mu_n \leq \l_{2n+1}$ everywhere in the real space $\M$.% Note however that in the complex neighborhood $\W$ of $\M$ in $\M^\C$ of Lemma \ref{wmungnlemma} the $\mu_n$'s no longer necessarily satisfy $\mu_n \in G_n$.

Conversely, according to van Moerbeke \cite{moer}, given any (real) Jacobi matrix $Q$ with spectrum $\l_1 < \l_2 \leq \l_3 < \l_4 \leq \l_5 < \ldots \l_{2N-2} \leq \l_{2N-1} < \l_{2N}$ and any sequence $(\mu_n)_{1 \leq n \leq N-1}$ with $\l_{2n} \leq \mu_n \leq \l_{2n+1}$ for $n=1, \ldots, N-1$, there are exactly $2^r$ $N$-periodic Jacobi matrices $Q$ with spectrum $(\l_n)_{1 \leq n \leq 2N}$ and Dirichlet spectrum $(\mu_n)_{1 \leq n \leq N-1}$, where $r$ is the number of $n$'s with $\l_{2n} < \mu_n < \l_{2n+1}$.

In the case where $(b,a) \in \M^\C$, we continue to define the Dirichlet eigenvalues $(\mu_n)_{1 \leq n \leq N-1}$ by (\ref{mundef}), and we list them in lexicographic ordering $\mu_1 \prec \mu_2 \prec  \ldots \prec \mu_{N-1}$.
%\begin{displaymath}
%  \mu_1 \prec \mu_2 \prec  \ldots \prec \mu_{N-1},
%\end{displaymath}
Then the $\mu_i$'s no longer depend continuously on $(b,a) \in \M^\C$. However, if we choose the complex neighborhood $\W$ of $\M$ in $\M^\C$ of Lemma \ref{wmungnlemma} small enough, then for any $(b,a) \in \W$ and $1 \leq n \leq N-1$, there exist isolating neighborhoods $(U_n)_{1 \leq n \leq N-1}$ so that $\mu_n$ is contained in the neighborhood $U_n$ of $G_n$ (but not necessarily in $G_n$ itself).

For later use, we compute the spectra of $Q(b,a)$ and $L_2(b,a)$ in the special case $(b,a) = (\b 1_N, \a 1_N)$ with $\b \in \R$ and $\a > 0$. Here $1_N$ denotes the vector $(1, \ldots, 1) \in \R^N$. These points are the equilibrium points (of the restrictions) of the Toda Hamiltonian vector field (to the symplectic leaves $\Mmba$). We compute the spectrum $(\l_j)_{1 \leq j \leq 2N}$ of the matrix $Q(\b 1_N, \a 1_N)$ and the Dirichlet eigenvalues $(\mu_l)_{1 \leq l \leq N-1}$ of $L = L(\b 1_N, \a 1_N)$. Furthermore, for any $1 \leq l \leq N-1$, we compute a normalized eigenvector corresponding to the eigenvalue $\mu_l$, $g_l = \big( g_l(j) \big)_{1 \leq j \leq N}$, i.e. $L g_l = \mu_l g_l$, $g_l(1)=0$, and a vector $h_l = \big( h_l(j) \big)_{1 \leq j \leq N}$ which is the normalized solution of $L y = \mu_l y$ orthogonal to $g_l$ satisfying $W(h_l, g_l)(N) > 0$.
\begin{lemma} \label{speclemma}
The spectrum $(\l_j)_{1 \leq j \leq 2N}$ of $Q(\b 1_N, \a 1_N \! )$ and the Dirichlet eigenvalues $(\mu_l)_{1 \leq l \leq N-1}$ of $L(\b 1_N, \a 1_N)$ are given by
\begin{eqnarray*}
  \l_1 & = & \b - 2 \a, \\
  \l_{2l} = \l_{2l+1} = \mu_l & = & \b - 2 \a \cos \frac{l \pi}{N} \quad (1 \leq l \leq N-1), \\
  \l_{2N} & = & \b +2 \a.
\end{eqnarray*}
In particular, all spectral gaps of $Q(\b 1_N, \a 1_N)$ are collapsed. For any $1 \leq l \leq N-1$, the vectors $g_l$ and $h_l$ defined by
\begin{eqnarray}
  g_l(j) & = & (-1)^{j+1} \sqrt\frac{2}{N} \sin\frac{(j-1)l\pi}{N} \quad (1 \leq j \leq N), \label{gkformula} \\
  h_l(j) & = & (-1)^j \sqrt\frac{2}{N} \cos\frac{(j-1)l\pi}{N} \quad (1 \leq j \leq N) \label{hkformula}
\end{eqnarray}
%with $d_l = \sqrt\frac{N}{2} \frac{1}{\sin \frac{l \pi}{N}}$ 
satisfy $L y = \mu_l y$ and the normalization conditions
\begin{displaymath}
  \sum_{j=1}^N g_l(j)^2 = \sum_{j=1}^N h_l(j)^2 = 1, \quad g_l(0) > 0, \quad g_l(1) = 0;
\end{displaymath}
\begin{displaymath}
  W(h_l, g_l)(N) > 0, \quad \langle h_l, g_l \rangle_{\R^N} = 0.
\end{displaymath}
\end{lemma}

\begin{remark}
Recall that $(g_l(j))_{1 \leq j \leq N}$ is a vector in $\R^N$. The normalization condition $g_l(0)>0$ means that $g_l(j) = \nu_l y_1(j,\mu_l)$ for $1 \leq j \leq N$ with $\nu_l > 0$.
\end{remark}

\section{Action-angle variables} \label{review}

In this section we summarize the results obtained in \cite{ahtk1} which we will need in the sequel. First we have to introduce some more notation.

\emph{Riemann surface $\Sigma_{b,a}$:} Denote by $\Sigma_{b,a}$ the Riemann surface obtained as the compactification of the affine curve $\mathcal{C}_{b,a}$ defined by
\begin{equation} \label{algcurve}
\{ (\l,z) \in \C^2 : z^2 = \Delta^2_{\l} (b, a) - 4 \}.% = \{ (\l,z) \in \C^2 : z^2 = \alpha^{-2N} \prod_{n=1}^{2N} (\l - \l_n) \}.
\end{equation}
Note that $\mathcal{C}_{b,a}$ and $\Sigma_{b,a}$ are spectral invariants. (Strictly speaking, $\Sigma_{b,a}$ is a Riemann surface only if the spectrum of $Q(b,a)$ is simple - see e.g. Appendix A in \cite{teschl2} for details in this case. If the spectrum of $Q(b,a)$ is \emph{not} simple, $\Sigma(b,a)$ becomes a Riemann surface after doubling the multiple eigenvalues - see e.g. section $2$ of \cite{kato}.)

\emph{Dirichlet divisors:} To the Dirichlet eigenvalue $\mu_n$ ($1 \leq n \leq N-1$) we associate the point $\mu_n^{*}$ on the surface $\Sigma_{b,a}$,
\begin{equation} \label{munstarred}
\mu_n^* := \left( \mu_n, \sqrt[*]{\Delta^2_{\mu_n} - 4} \right) \; \textrm{with} \;\; \sqrt[*]{\Delta^2_{\mu_n} - 4} := y_1(N, \mu_n) - y_2(N+1, \mu_n),
\end{equation}
where we used that $\Delta^2_{\mu_n} - 4 = \left( y_1(N, \mu_n) - y_2(N+1, \mu_n) \right)^2$.

\emph{Standard root:} The standard root or $s$-root for short, $\sqrt[s]{1 - \l^2}$, is defined for $\l \in \C \setminus [-1,1]$ by
\begin{equation} \label{sroot}
  \sqrt[s]{1 - \l^2} := i \l \sqrt[+]{1 - \l^{-2}}.
\end{equation}
More generally, we define for $\l \in \C \setminus \{ ta+(1-t)b \, | \, 0 \leq t \leq 1 \}$ the $s$-root of a radicand of the form $(b - \l)(\l - a)$ with $a \prec b, a \neq b$ by
\begin{equation} \label{sroot2}
\sqrt[s]{(b - \l)(\l - a)} := \frac{\gamma}{2} \sqrt[s]{1 - w^2},
\end{equation}
where $\gamma := b-a$, $\tau := \frac{b+a}{2}$ and $w := \frac{\l - \tau}{\gamma/2}$.

\emph{Canonical sheet and canonical root:} For $(b,a) \in \M$ the canonical sheet of $\Sigma_{b,a}$ is given by the set of points $(\l, \sqrt[c]{\Delta_\l^2 - 4})$ in $\mathcal{C}_{b,a}$, where the $c$-root $\sqrt[c]{\Delta_\l^2 - 4}$ is defined on $\C \setminus \bigcup_{n=0}^{N} (\l_{2n}, \l_{2n+1})$ (where $\l_0 := -\infty$ and $\l_{2N+1} := \infty$) and determined by the sign condition
\begin{equation} \label{croot}
-i \sqrt[c]{\Delta_\l^2 - 4} > 0 \quad \textrm{for} \quad \l_{2N-1} < \l < \l_{2N}.
\end{equation}
As a consequence one has for any $1 \leq n \leq N$
\begin{equation} \label{croot2}
\textrm{sign} \; \sqrt[c]{\Delta_{\l - i0}^2 - 4} = (-1)^{N+n-1} \quad \textrm{for} \quad \l_{2n} < \l < \l_{2n+1}.
\end{equation}
The definition of the canonical sheet and the $c$-root can be extended to the neighborhood $\W$ of $\M$ in $\M^\C$ of Lemma \ref{wmungnlemma}.

\emph{Abelian differentials:} Let $(b,a) \in \M$ and $1 \leq n \leq N-1$. Then there exists a unique polynomial $\psi_n(\l)$ %Assume that for a given element $(b,a) \in \M$ the spectrum of $Q(b,a)$ is simple. Then
%the $N-1$ differentials $\frac{z^j dz}{\sqrt{\Delta^2_\l-4}}$ $(0 \leq j \leq N-2)$ are a basis of the space of holomorphic differentials on the surface $\Sigma_{b,a}$ (see e.g. \cite{fakr}). Hence
%$\Sigma_{b,a}$ admits a basis of holomorphic differentials of the form $\frac{\psi_n(\l)}{\sqrt{\Delta^2_\l-4}} d\l$ $(1 \leq n \leq N-1)$ where the $\psi_n(\l)$ are polynomials 
of degree at most $N-2$ such that for any $1 \leq k \leq N-1$
\begin{equation} \label{psi}
\frac{1}{2\pi} \int_{c_k} \frac{\psi_n(\l)}{\sqrt{\Delta^2_\l-4}} \, d\l = \delta_{kn}.
\end{equation}
Here, for any $1 \leq k \leq N-1$, $c_k$ denotes the lift of the contour $\Gamma_k$ to the canonical sheet of $\Sigma_{b,a}$. For any $k \neq n$ with $\l_{2k} \neq \l_{2k+1}$, it follows from (\ref{psi}) that
\begin{equation} \label{psiproperty}
\frac{1}{\pi} \int_{\l_{2k}}^{\l_{2k+1}} \frac{\psi_n(\l)}{\sqrt[+]{\Delta^2_\l - 4}} \, d\l = 0.
\end{equation}
Hence in every gap $(\l_{2k}, \l_{2k+1})$ with $k \neq n$ the polynomial $\psi_n$ has a zero which we denote by $\sigma_k^n$. %If the spectrum of $Q(b,a)$ is not simple, i.e. if $\gamma_k = 0$ for some $1 \leq k \leq N-1$, it can be shown (as explained e.g. in \cite{kato} for KdV) that (\ref{psi}) still holds, and 
If $\l_{2k} = \l_{2k+1}$ then it follows from (\ref{psi}) and Cauchy's theorem that $\sigma_k^n = \l_{2k} = \l_{2k+1}$. As $\psi_n(\l)$ is a polynomial of degree at most $N-2$, one has
\begin{equation} \label{psiprodrepr}
  \psi_n(\l) = M_n \prod_{1 \leq k \leq N-1 \atop k \neq n} (\l - \sigma_k^n),
\end{equation}
%with nonzero constants $(K_n(b,a))_{1 \leq n \leq N-1}$ and zeroes $(\sigma_k^n)_{1 \leq k \leq N-1, k \neq n}$ satisfying $\l_{2k} \leq \sigma_k^n \leq \l_{2k+1}$.
where $M_n \equiv M_n(b,a) \neq 0$. % (In fact, $M_n(b,a)$ can be computed to be $2 (-1)^N s_n \a^{-(N-1)}$.)

%In particular it follows that $(-1)^{N+n+1} \psi_n(\l) > 0$ on $[\l_{2n}, \l_{2n+1}]$.

In a straightforward way one can prove that there exists a neighborhood $\W$ of $\M$ in $\M_{\C}$, so that for any $(b,a) \in \W$ and any $1 \leq n \leq N-1$, there is a unique polynomial $\psi_n(\l)$ of degree at most $N-2$ satisfying (\ref{psi}) for any $1 \leq k \leq N-1$ as well as the product representation (\ref{psiprodrepr}), and so that the zeroes are analytic functions on $\W$.

%In this section, we define candidates for the Birkhoff coordinates on $\M$ and investigate their analyticity properties. 
We have seen in the introduction that there are two Casimir functions $C_1$ and $C_2$ for $J$, leading to the symplectic foliation of $\M$ with the leaves $\Mba$. In \cite{ahtk1} we defined global action variables $(I_n)_{1 \leq n \leq N-1}$ on $\M$ and, for any $1 \leq n \leq N-1$, the angle variable $\theta_n$ on $\M \setminus D_n$ where $D_n$ is given by (\ref{dndef}) and (\ref{dncdef}).
\begin{definition} \label{actionsdef}
Let $(b,a) \in \M$. For $1 \leq n \leq N-1$,
\begin{equation} \label{actsba}
I_n := \frac{1}{2\pi} \int_{\Gamma_n} \l \frac{\dot{\Delta}_\l}{\sqrt[c]{\Delta^2_\l-4}} \; d\l
\end{equation}
where $\dot{\Delta}_\l = \frac{d}{d\l} \Delta_\l$ is the $\l$-derivative of the discriminant $\Delta_\l = \Delta(\l, b, a)$ and the contour $\Gamma_n$ and the canonical root $\sqrt[c]{\cdot}$ are given as in section \ref{tools}.
\end{definition}

\begin{definition} \label{angledefinition}
For any $1 \leq n \leq N-1$, the function $\theta_n$ is defined for $(b,a) \in \M \setminus D_n$ by
\begin{equation} \label{angle1}
\theta_n := \eta_n + \beta_n \;\; (\textrm{mod} \; 2 \pi) \quad \textrm{and} \quad \beta_n := \sum_{n \neq k=1}^{N-1} \beta_k^n,
\end{equation}
where for $k \neq n$,
\begin{equation} \label{angle2}
\beta_k^n = \int_{\l_{2k}}^{\mu_k^*} \! \frac{\psi_n(\l)}{\sqrt{\Delta^2_\l-4}} \, d\l, \quad \eta_n = \int_{\l_{2n}}^{\mu_n^*} \! \frac{\psi_n(\l)}{\sqrt{\Delta^2_\l-4}} \, d\l \;\; (\textrm{mod} \; 2\pi).
\end{equation}
Here for any $1 \leq k \leq N-1$, $\mu_k^*$ is the Dirichlet divisor defined in (\ref{munstarred}), and $\l_{2k}$ is identified with the ramification point $(\l_{2k}, 0)$ on the Riemann surface $\Sigma_{b,a}$. The integration paths on $\Sigma_{b,a}$ in (\ref{angle2}) are required to be admissible in the sense that their image under the projection $\pi: \Sigma_{b,a} \to \C$ on the first component stays inside the isolating neighborhoods $U_k$.
\end{definition}

In \cite{ahtk1} we proved the following results.

\begin{theorem} \label{analytic}
  \begin{itemize}
  \item[(A)] There exists a complex neighborhood $\W$ of $\M$ in $\M^\C$ with the following properties:
\begin{itemize}
\item[(i)] For any $1 \leq n \leq N-1$, the functions $I_n: \W \to \C$, $\theta_n: \W \setminus D_n \to \C \, (\textrm{mod } \pi)$, and $\b_n: \W \to \C$ are analytic.
\item[(ii)] On the real space $\M$, each function $I_n$ is real-valued and nonnegative. It vanishes at a point $(b,a) \in \M$ if and only if the $n$-th gap is collapsed, i.e. if $\gamma_n(b,a) = 0$. Moreover $\b_n(b,a) = 0$ for any $(b,a) \in \M$ with $\gamma_n(b,a) = 0$.
\item[(iii)] For any $1 \leq n \leq N-1$, the quotient $I_n / \gamma_n^2$ extends analytically from $\M \setminus D_n$ to all of $\W$ and has strictly positive real part on $\W$. As a consequence, $\xi_n = \sqrt[+]{2 I_n / \gamma_n^2}$ is a well-defined, analytic and nonvanishing function on $\W$, where$\sqrt[+]\cdot$ is the principal branch of the square root on $\C \setminus (-\infty, 0]$.
\end{itemize}

\item[(B)] The variables $I_n$ and $\theta_n$, $1 \leq n \leq N-1$, are globally defined action-angle variables for the periodic Toda lattice. More precisely:
\begin{itemize}
\item[(iv)] The functions $(I_n)_{1 \leq n \leq N-1}$ are pairwise in involution and Poisson commute with the Toda Hamiltonian $H$, i.e. for any $1 \leq m,n \leq N-1$, $i = 1,2$,
  \begin{displaymath}
    \{ I_m, I_n \}_J = 0, \quad \{ H, I_n \}_J = 0 \quad \textrm{and} \quad \{ C_i, I_n \}_J = 0 \quad \textrm{on } \W.
  \end{displaymath}
\item[(v)] The functions $\theta_n: \W \setminus D_n \to \R$, $1 \leq n \leq N-1$, are conjugate to the variables $(I_m)_{1 \leq m \leq N-1}$, i.e. for any $1 \leq m \leq N-1$, $j=1,2$,
  \begin{displaymath}
    \{ I_m, \theta_n \}_J = \delta_{mn} \quad \textrm{and} \quad \{ C_i, \theta_n \}_J = 0 \quad \textrm{on } \W \setminus D_n
  \end{displaymath}
and
\begin{displaymath}
  \{ \theta_m, \theta_n \}_J = 0 \quad \textrm{on } \W \setminus (D_m \cup D_n).
\end{displaymath}
\end{itemize}
  \end{itemize}
\end{theorem}

\section{Birkhoff map} \label{birkhoff}

In this section, we construct the map $\Omega$, $\Omega = ((\Omega_n)_{1 \leq n \leq N-1}, C_1, C_2)$, defined on $\W$. For any $1 \leq n \leq N-1$ and $(b,a) \in \W \setminus D_n$ the $n$-th component of $\Omega$, $\Omega_n = (x_n,y_n)$, is defined to be
\begin{displaymath}
(x_n,y_n) = \sqrt[+]{2 I_n} (\cos \theta_n, \sin \theta_n).
\end{displaymath}
%corresponding to $I_n$, $\theta_n$.

%For any $1 \leq n \leq N-1$ and $(b,a) \in \M \setminus D_n$, we define the Cartesian coordinates by $x_n := \sqrt[+]{2 I_n} \cos \theta_n$ and $y_n := \sqrt[+]{2 I_n} \sin \theta_n$. 
In order to extend $(x_n, y_n)$ to all of $\W$, we substitute the formula $2 I_n = \gamma_n^2 \xi_n^2$ of Theorem \ref{analytic} into the definition of $(x_n, y_n)$. Hence, for $(b,a) \in \W \setminus D_n$,
\begin{displaymath}
\left\{ \begin{array}{ccccc} x_n & = & \xi_n \gamma_n \cos \theta_n & = & \xi_n \gamma_n \frac{e^{i\theta_n} + e^{-i\theta_n}}{2} \\
y_n & = & \xi_n \gamma_n \sin \theta_n & = & \xi_n \gamma_n \frac{e^{i\theta_n} - e^{-i\theta_n}}{2i}.
\end{array} \right.
\end{displaymath}

%Recall that $\theta_n = \eta_n + \b_n$, where $\b_n = \sum_{n \neq k=1}^{N-1} \b_k^n$. It is not hard to see that $\b_n$ is defined and analytic on all of $M$, including $D_n$. Thus it remains to analyze
By Theorem \ref{analytic}, $\b_n$ - and therefore $e^{\pm i \b_n}$ - as well as $\xi_n$ are analytic on $\W$. Thus it remains to analytically extend the functions
\begin{equation} \label{zndef}
z_n^{\pm} := \gamma_n e^{\pm i \eta_n}
\end{equation}
to $\W$. Note that at this point, $z_n^\pm$ is defined on $\W \setminus D_n$ only. The following result is proved in section \ref{proofznpm} below.

\begin{prop} \label{znpm}
The functions $z_n^{\pm}$ extend analytically to $\W$. On $\M \cap D_n$, $z_n^{\pm} = 0$.
\end{prop}

\begin{definition} \label{coordgendef}
For $(b,a) \in \W$ and $1 \leq n \leq N-1$,
\begin{equation} \label{cartesianfinal}
\left\{ \begin{array}{ccc}
x_n & := & \frac{\xi_n}{2} (z_n^ + e^{i\b_n} + z_n^- e^{-i\b_n}) \\
y_n & := & \frac{\xi_n}{2i} (z_n^ + e^{i\b_n} - z_n^- e^{-i\b_n})
\end{array} \right.
\end{equation}
\end{definition}

For any $1 \leq n \leq N-1$, it follows from (\ref{cartesianfinal}) that% $x_n + iy_n = \xi_n z_n^+ e^{i \b_n}$.
\begin{displaymath}
x_n \pm iy_n = \xi_n z_n^\pm e^{i \b_n}.
\end{displaymath}

Now we are ready to define the coordinate map $\Omega$. As the Casimir function $C_2$ takes only positive values, we introduce as target space of $\Omega$ the model space
$$ \P := \R^{2(N-1)} \times \R \times \R_{>0} $$
and define
%\begin{definition}
%The map $\Omega$ is defined by
\begin{equation} \label{coordinates}
\begin{array}{ccll}
 \Omega: & \M & \to & \P \\
 & (b,a) & \mapsto & ((x_n, y_n)_{1 \leq n \leq N-1}, C_1, C_2).
\end{array}
\end{equation}
%\end{definition}

In view of Theorem \ref{analytic} and Proposition \ref{znpm} we then have proved
\begin{theorem} \label{coordanaly}
The map $\Omega: \M \to \P$ is real analytic. It extends analytically to the complex neighborhood $\W$ of $\M$ in $\M^\C$ of Theorem \ref{analytic}.
\end{theorem}

To compute the differential of the map $\Omega$, we first compute for any $1 \leq n \leq N-1$ the gradient $\nba z_n^+$ of $z_n^+$ for $1 \leq n \leq N-1$. Let us first recall some notation introduced in \cite{ahtk1}. For sequences $\big( v(j)_{j \in \Z} \big), \big( w(j)_{j \in \Z} \big) \subseteq \C$ define the $N$-vectors
\begin{eqnarray}
v \cdot w & := & \big( v(j) w(j) \big)_{1 \leq j \leq N}, \label{nprod1}\\
v \cdot S  w & := & \big( v(j) w(j+1) \big)_{1 \leq j \leq N}, \label{nprod2}
\end{eqnarray}
where $S$ denotes the shift operator of order $1$. % For $v = w$ we write
%\begin{displaymath}
%v^2 := v \cdot v.
%\end{displaymath}
Further define the $2N$-vector
\begin{equation} \label{vector2n}
v \cds w := (v \cdot w, v \cdot Sw + w \cdot Sv).
\end{equation}
In case $v = w$ we also use the shorter notation
\begin{equation} \label{vector2nvv}
v^\mbf2 := v \cds v.
\end{equation}
Written componentwise, $v \cds w$ is the $2N$-vector
\setlength\arraycolsep{0.1pt} {
\begin{displaymath} \!\!\!\!\!\ (v \cds w)(j) \! = \!\!
\left\{ \begin{array}{cc}
v(j) w(j) & (1 \leq j \leq N) \\
v(j-N) w(j-N+1) + v(j-N+1) w(j-N) & (N \! < \! j \leq \! 2N) \\
\end{array} \right..
\end{displaymath}}

The following proposition will be proved in section \ref{proofdbazn}.
\begin{prop} \label{dbaznnew}
At any point $(b,a) \in \M \cap D_n$, the gradient $\nba z_n^+$ is given by
\begin{equation} \label{dbaznformulaapp}
\nba z_n^+ \equiv (\n_b z_n^+, \n_a z_n^+) = (h_n - i g_n)^\mbf2, %\big( (h_n + i g_n)^2, 2 (h_n + i g_n) \cdot S(h_n + i g_n) \big),
\end{equation}
where $g_n$ and $h_n$ are defined by (\ref{gkformula}) and (\ref{hkformula}), respectively.
\end{prop}

It is convenient to introduce the complex version of $\Omega$,
\begin{equation} \label{gbcdef}
\begin{array}{ccll}
 \Omega^\C: & (\M, J) & \to & (\C^{N-1} \times \R \times\R_{>0}, J_0) \\
 & (b,a) & \mapsto & ((x_n + i y_n)_{1 \leq n \leq N-1}, C_1, C_2),
\end{array}
\end{equation}
and the abbreviations $s_n := \sin \frac{n \pi}{N}$ for $1 \leq n \leq N-1$. Moreover, for the rest of this section we write $\sqrt{\cdot}$ instead of $\sqrt[+]{\cdot}$ for the principal branch of the square root function, defined on $\C \setminus (-\infty, 0]$. 
\begin{prop} \label{birkdiff}
Let $\b \in \R$ and $\a > 0$. The gradient $\nba \Omega^\C$ of $\Omega^\C$ at $(b, a) = (\b 1_N, \a 1_N)$ is given by ($1 \leq n \leq N-1$)
\begin{equation} \label{gradxnyn}
(\nba x_n + i \, \nba y_n)(j) = \frac{1}{\sqrt{2 \a N}} \cdot \frac{1}{\sqrt{s_n}} \left\{  \begin{array}{ll}
e^\frac{(2j-2) i \pi n}{N} & (1 \leq j \leq N) \\
-2 e^\frac{(2j-1) i \pi n}{N} & (N+1 \leq j \leq 2N)
\end{array} \right.
\end{equation}
and
\begin{equation} \label{casgrad}
\nba C_1 = -\frac{1}{N}(1_N, 0_N), \quad \nba C_2 = \frac{1}{N}(0_N, 1_N),
\end{equation}
where $0_N = (0, \ldots, 0) = 0 \cdot 1_N$.
\end{prop}

\begin{proof}%[Proof of Theorem \ref{birkdiff}]
Note that the formulas in (\ref{casgrad}) immediately follow from the formulas (\ref{c1grad}) and (\ref{c2grad}) for the gradients of the Casimir functions $C_1$ and $C_2$ at an arbitrary point $(b,a) \in \M$. It remains to show (\ref{gradxnyn}). In order to compute the gradient of $x_n + iy_n$, we use the formula $x_n + iy_n = \xi_n z_n^+ e^{i \b_n}$ stated above. By Lemma \ref{speclemma}, all gaps are closed for $(b,a) = (\b 1_N, \a 1_N)$. Hence, by Theorem \ref{analytic} (ii), $\b_n = 0$ and thus $e^{i \b_n} = 1$ for any $1 \leq n \leq N-1$. In addition, by Proposition \ref{znpm}, $z_n^+ = 0$. The gradient of $x_n + iy_n = \xi_n z_n^+ e^{i \b_n}$ with respect to $(b,a)$ at $(b, a) = (\b 1_N, \a 1_N)$ is therefore given by
\begin{equation} \label{xkykgrad}
\nba x_n + i \, \nba y_n = \xi_n \nba z_n^+.
\end{equation}

Further, by Theorem \ref{analytic},
\begin{equation} \label{xiformula}
  \xi_n = \lim_{\gamma_n \to 0} \frac{\sqrt{2I_n}}{\gamma_n}.
\end{equation}
The proof in \cite{ahtk1} of the results stated in Theorem \ref{analytic} shows that
\begin{displaymath}
\xi_n = \frac{\sqrt{N}}{2} \sqrt{\chi_n(\tau_n)},
\end{displaymath}
where
\begin{equation} \label{chidef}
\chi_n(\l) = \frac{(-1)^{N-n-1}}{\sqrt[+]{(\l-\l_1)(\l_{2N}-\l)}} \prod_{m \neq n} \frac{\l-\dot{\l}_m} {\sqrt[+]{(\l-\l_{2m+1})(\l-\l_{2m})}},
\end{equation}
and
\begin{equation} \label{taudef}
  \tau_n = \frac{1}{2} (\l_{2n} + \l_{2n+1})
\end{equation}
As all gaps are collapsed, $\dot{\l}_n = \l_{2n} = \l_{2n+1} = \tau_n$ for any $1 \leq n \leq N-1$. Hence, for $\l = \tau_n$, the product in (\ref{chidef}) is equal to $(-1)^{N-n-1}$ and thus
\begin{displaymath}
\xi_n = \frac{\sqrt{N}}{2} \sqrt{\chi_n(\tau_n)} = \frac{\sqrt{N}}{2} \big( (\l_{2N} - \tau_n)(\tau_n- \l_1) \big)^{-\frac{1}{4}}.
\end{displaymath}
By Lemma \ref{speclemma}, $\l_1 = \b - 2\a$, $\l_{2N} = \b + 2\a$, and $\tau_n = \b - 2\a \cos\frac{n\pi}{N}$. Therefore
\begin{equation} \label{xikformula}
  \xi_n = \frac{\sqrt{N}}{2} \left( 4\a^2 \left( 1 - \cos^2\frac{n \pi}{N} \right) \right)^{-\frac{1}{4}} = \left( \frac{8\a}{N} \sin\frac{n \pi}{N} \right)^{-\frac{1}{2}} = \sqrt{\frac{N}{8\a}} \cdot \frac{1}{\sqrt{s_n}}.
\end{equation}
Next, by Proposition \ref{dbaznnew}, the gradient $\nba z_n^+$ of $z_n^+$ in (\ref{xkykgrad}) is given by
\begin{equation} \label{dzkres}
\nba z_n^+ = (h_n - i g_n)^\mbf2
\end{equation}
where we used the notation introduced in (\ref{vector2nvv}). From the formulas (\ref{gkformula})-(\ref{hkformula}) for $g_n$ and $h_n$ we then obtain from (\ref{dzkres}) in the case $(b,a) = (\b 1_N, \a 1_N)$
\begin{equation} \label{dzkformula}
 (\nba z_n^+)(j) = \frac{2}{N} \left\{  \begin{array}{ll}
e^\frac{(2j-2)i \pi n}{N} & (1 \leq j \leq N) \\
-2 e^\frac{(2j-1)i \pi n}{N} & (N+1 \leq j \leq 2N)
\end{array} \right..
\end{equation}
Substituting (\ref{xikformula}) and (\ref{dzkformula}) into (\ref{xkykgrad}) then yields the claimed formula (\ref{gradxnyn}) and therefore completes the proof of Theorem \ref{birkdiff}.
\end{proof}

We end this section with proving commutator relations among the variables $(x_n, y_n)_{1 \leq n \leq N-1}$ which will be used in section \ref{diffeochapter} to prove Theorem \ref{sumthm}.

\begin{prop} \label{canrel}
For any $(b,a) \in \W$ and $1 \leq k,l \leq N-1$, the following relations hold:
$$ \{ x_k, x_l \}_J = 0 \,; \; \{ y_k, y_l \}_J = 0 \, ; \; \{ x_k, y_l \}_J = \delta_{kl}. $$
%\begin{eqnarray*}
%  \{ x_k, x_l \} & = & 0 \\
%  \{ y_k, y_l \} & = & 0 \\
%  \{ x_k, y_l \} & = & \delta_{kl}.
%\end{eqnarray*}
\end{prop}

\begin{proof}
By the continuity of $\{ \cdot, \cdot \}_J$, it is sufficient to prove the claimed relations for any $1 \leq k,l \leq N-1$ and $(b,a) \in \W \setminus (D_k \cup D_l)$. In this case, $x_n = \sqrt{2 I_n} \cos \theta_n$ and $y_n = \sqrt{2 I_n} \sin \theta_n$ for $n \in \{ k,l \}$. Let us first show $\{ x_k, y_l \} = \delta_{kl}$.
\setlength\arraycolsep{2pt}{
\begin{eqnarray*}
  \{ x_k, y_l \}_J & = & \{ \sqrt{2 I_k} \cos \theta_k, \sqrt{2 I_l} \sin \theta_l \}_J \\
 & = & \cos \theta_k \sqrt{2 I_l} \{ \sqrt{2 I_k}, \sin \theta_l \}_J + \sin \theta_l \sqrt{2 I_k} \{ \cos \theta_k, \sqrt{2 I_l} \}_J \\
 & = & \cos \theta_k \cos \theta_l \frac{\sqrt{2 I_l}}{\sqrt{2 I_k}} \{ I_k, \theta_l \}_J +  \sin \theta_k \sin \theta_l \frac{\sqrt{2 I_k}}{\sqrt{2 I_l}} \{ I_l, \theta_k \}_J \\
 & = & \delta_{kl},
\end{eqnarray*}
where for the latter identity we used Theorem \ref{analytic} (B). The other two claimed relations are proved similarly.}
\end{proof}

\section{Proof of Proposition \ref{znpm}} \label{proofznpm}

To prove Proposition \ref{znpm} we follow the arguments used in \cite{kapo} to prove a similar result for KdV.

Recall from (\ref{angle2}) that $\eta_n$ is the following integral on $\Sigma_{b,a}$,
\begin{displaymath}
\eta_n = \int_{\l_{2n}}^{\mu_n^*} \frac{\psi_n(\l)}{\sqrt{\Delta^2_\l - 4}} \, d\l \quad (\textrm{mod} \; 2 \pi),
\end{displaymath}
where $\mu_n^*$ is the Dirichlet divisor introduced in (\ref{munstarred}), and $\l_{2n}$ is identified with the ramification point $(\l_{2n}, 0)$ on $\Sigma_{b,a}$.

Note that on $\W \setminus D_n$, $z_n^\pm = \gamma_n e^{\pm i \eta_n}$ is continuous. Indeed, possible discontinuities of $\eta_n$ due to the lexicographic ordering of the eigenvalues $(\l_j)_{1 \leq j \leq 2N}$ lead simultaneously to a sign change of $\gamma_n$ \emph{and} $e^{\pm i \eta_n}$, thus leaving $\gamma_n e^{\pm i \eta_n}$ unaffected. For $\l$ near the interval $G_n$, defined in (\ref{gndef}), one has
\begin{equation} \label{psiexp}
\frac{\psi_n(\l)}{\sqrt[c]{\Delta^2_\l - 4}} = \frac{\zeta_n(\l)}{\sqrt[s]{(\l_{2n+1} - \l)(\l - \l_{2n})}},
\end{equation}
where% $\zeta_n$ is defined by
\begin{equation} \label{zetafndef}
  \zeta_n(\l) := \frac{M_n'}{\sqrt[+]{(\l - \l_1)(\l_{2N} - \l)}} \prod_{m \neq n} \frac{\l - \sigma_m^n}{\sqrt[+]{(\l_{2m+1} - \l)(\l_{2m} - \l)}},
\end{equation}
with $M_n' \neq 0$. Note that $\zeta_n$ is analytic and nonvanishing in $U_n$. %In addition, $\zeta_n$ is nonnegative for real $(b,a)$.
We claim that
\begin{equation} \label{zetagamma}
\zeta_n(\mu) = 1 + O(|\gamma_n|)
\end{equation}
for $\mu \in G_n$, locally uniformly on $\W$. Indeed, for real $(b,a)$ with $\gamma_n > 0$ and any $\mu \in G_n$ we deduce from (\ref{psi}), using that on the interval $(\l_{2n}, \l_{2n+1})$, both $(-1)^{N+n+1} \psi_n(\l)$ and $\zeta_n(\l)$ are positive,
\begin{eqnarray*}
  \pi \; & = \; & \int_{\l_{2n}}^{\l_{2n+1}} \frac{(-1)^{N+n+1} \psi_n(\l)}{\sqrt[+]{\Delta^2_\l - 4}} \, d\l \\
& = & \int_{\l_{2n}}^{\l_{2n+1}} \frac{\zeta_n(\mu) + \big( \zeta_n(\l) - \zeta_n(\mu) \big)}{\sqrt[+]{(\l_{2n+1} - \l)(\l - \l_{2n})}} \, d\l \\
& = \; & \pi \zeta_n(\mu) + O \big( \sup_{\l \in G_n} |\zeta_n(\l) - \zeta_n(\mu)| \big),
\end{eqnarray*}
where we used that
\begin{displaymath}
  \int_{\l_{2n}}^{\l_{2n+1}} \frac{d\l}{\sqrt[+]{(\l_{2n+1} - \l)(\l - \l_{2n})}} = \pi.
\end{displaymath}
Hence for $\mu \in G_n$,
$$ \zeta_n(\mu) = 1 + O \big( \sup_{\l \in G_n} |\zeta_n(\l) - \zeta_n(\mu)| \big). $$
By Cauchy's estimate, $\sup_{\l, \mu \in G_n}|\zeta_n(\l) - \zeta_n(\mu)| \leq M|\gamma_n|$, where $M$ can be chosen locally uniformly on $\W$. This proves the claimed estimate (\ref{zetagamma}) for real $(b,a)$. For complex $(b,a) \in \W$, the preceding identities remain true at least up to a sign. By the continuity of $\zeta_n$ in $(b,a)$ and $\l$, the estimate (\ref{zetagamma}) remains valid on $\W$.

We now investigate the limiting behavior of $z_n^\pm$ as the $n$-th gap collapses. This limit exists and does not vanish when $(b,a)$ is in the open set
$$ X_n := \{ (b,a) \in \W: \mu_n(b,a) \notin G_n(b,a) \}. $$
Note that $X_n$ does not intersect the real space $\M$, since $\mu_n \in [\l_{2n}, \l_{2n+1}]$ for real $(b,a)$.

We now define
\begin{equation} \label{chinbadef}
\chi_n(b,a) := \int_{\tau_n}^{\mu_n} \frac{\zeta_n(\l) - \zeta_n(\tau_n)}{\l - \tau_n} d\l,
\end{equation}
with $\tau_n = (\l_{2n} + \l_{2n+1})/2$. Note that $\tau_n$ is analytic on $\W$. Indeed, using the product representation (\ref{delta2lrepr}) of $\Delta_\l^2 - 4$ one gets by the residue theorem
\begin{equation} \label{tauanalproof}
  \tau_n = \frac{1}{2 \pi i} \int_{\Gamma_n} \l \frac{\Delta_\l \dot{\Delta}_\l}{\Delta^2_\l - 4} d\l.
\end{equation}
Since, locally on $\M$, the contour $\Gamma_n$ can be kept fixed and $\Delta_{\l}(b,a)$ is analytic on $\C \times \W$, (\ref{tauanalproof}) shows that $\tau_n$ is a real analytic function on $\W$.

 As $\mu_n$ and $\zeta_n$ are analytic on $\C \times \W$, it then follows that $\chi_n$, defined by (\ref{chinbadef}), is analytic on $\W$.

To facilitate the statement of the following result, define, for any $1 \leq n \leq N-1$, the sign $\epsilon_n = \pm 1$ for elements $(b,a)$ in $X_n$ so that
\begin{equation} \label{defeps}
\frac{\psi_n(\mu_n)}{\sqrt[*]{\Delta^2_{\mu_n} - 4}} = \frac{\epsilon_n \zeta_n(\mu_n)}{\sqrt[s]{(\l_{2n+1} - \mu_n)(\mu_n - \l_{2n})}}.
\end{equation}
Note that the $s$-root is well defined, since $\mu_n \notin G_n$ for $(b,a) \in X_n$. To prove Proposition \ref{znpm} we need the following auxiliary result:

\begin{lemma} \label{applemma}
As $(b,a) \in \W \setminus D_n$ tends to $(b_0, a_0) \in D_n \cap X_n$,
$$ \gamma_n e^{\pm i \eta_n} \to -2 (1 \pm \epsilon_n)(\mu_n - \tau_n) e^{\pm \epsilon_n \chi_n}, $$
where $\epsilon_n$ is defined by (\ref{defeps}).
\end{lemma}

\begin{proof}
Since $X_n$ is open and $(b_0, a_0) \in X_n \cap D_n$, it follows that $(b,a) \in X_n$ for all $(b,a)$ sufficiently close to $(b_0, a_0)$. Also, $(b,a) \notin D_n$ by assumption.

For $(b,a) \in X_n \setminus D_n$ one has, modulo $2 \pi$,
\begin{eqnarray*}
  \eta_n & = & \int_{\l_{2n}}^{\mu_n^*} \frac{\psi_n(\l)}{\sqrt{\Delta^2_\l - 4}} \, d\l \\
& = & \epsilon_n \int_{\l_{2n}}^{\mu_n} \frac{\zeta_n(\l)}{\sqrt[s]{(\l_{2n+1} - \l)(\l - \l_{2n})}} \,d\l \\
& = & \epsilon_n \int_{\l_{2n}}^{\mu_n} \frac{\zeta_n(\l_{2n})}{\sqrt[s]{(\l_{2n+1} - \l)(\l - \l_{2n})}} \,d\l + \epsilon_n \int_{\l_{2n}}^{\mu_n} \frac{\zeta_n(\l) - \zeta_n(\l_{2n})}{\sqrt[s]{(\l_{2n+1} - \l)(\l - \l_{2n})}} \,d\l \\
& = & \eta_n^{(1)} + \eta_n^{(2)} \quad \textrm{mod} \; 2\pi.
\end{eqnarray*}
The limiting behavior of the second term $\eta_n^{(2)}$ is straightforward. If $(b,a) \to (b_0, a_0)$, then $\gamma_n \to 0$ and so for $\l \neq \tau_n(b_0, a_0)$,
$$ \sqrt[s]{(\l_{2n+1} - \l)(\l - \l_{2n})} \to i (\l - \tau_n) $$
by the definition of the $s$-root. Hence, by the definition of $\chi_n$,
$$ i \eta_n^{(2)} \to \epsilon_n \int_{\tau_n}^{\mu_n} \frac{\zeta_n(\l) - \zeta_n(\tau_n)}{\l - \tau_n} d\l = \epsilon_n \chi_n. $$
Consequently, as $(b,a) \to (b_0, a_0)$,
$$ e^{i \eta_n^{(2)}} \to e^{\epsilon_n \chi_n}. $$
Turning to $\eta_n^{(1)}$, make the substitution $\l = \tau_n + z \gamma_n/2$. Then, by the definition (\ref{sroot2}) of the $s$-root,
\begin{equation} \label{phiexprlambda}
  \int_{\l_{2n}}^{\mu_n} \frac{d\l}{\sqrt[s]{(\l_{2n+1} - \l)(\l - \l_{2n})}} = \int_{-1}^{\rho_n} \frac{dz}{\sqrt[s]{1 - z^2}} = \phi(\rho_n)
\end{equation}
with
\begin{equation} \label{rhodefine}
  \rho_n = \frac{\mu_n - \tau_n}{\gamma_n/2}, \quad \phi(w) := \int_{-1}^w \frac{dz}{\sqrt[s]{1 - z^2}}.
\end{equation}
It follows that
\begin{equation} \label{expphiformula}
e^{i \phi(w)} = -w + i \sqrt[s]{1 - w^2},
\end{equation}
as both sides of the latter identity are analytic, univalent functions on $\C \setminus [-1, 1]$, which have the same limit at $-1$ and satisfy the same differential equation $\frac{f'(w)}{f(w)} = \frac{i}{\sqrt[s]{1 - w^2}}$. Hence, writing
\begin{displaymath}
  \exp (i \eta_n^{(1)}) = \exp \big( i \epsilon_n \phi(\rho_n) \zeta_n(\l_{2n}) \big) = \exp \big( i \phi(\rho_n) \epsilon_n \big) \exp \big( i \epsilon_n \phi(\rho_n) \hat{\zeta}_n \big)
\end{displaymath}
with $\hat{\zeta}_n = \zeta_n(\l_{2n}) - 1$, we obtain for $(b,a) \in X_n \setminus D_n$
\begin{equation} \label{hatexpr}
\gamma_n e^{i \eta_n^{(1)}} = \gamma_n (-\rho_n + i \sqrt[s]{1 - \rho_n^2})^{\epsilon_n} \cdot e^{i \epsilon_n \phi(\rho_n) \hat{\zeta}_n}.
\end{equation}
Passing to the limit $(b,a) \to (b_0, a_0)$, we have $\gamma_n \to 0$, while $v_n := \mu_n - \tau_n$ tends to a limit different from zero, and hence $|\rho_n| \to \infty$. By the definition (\ref{sroot}) of the $s$-root, the limit of the first two factors on the right hand side of the above equation can then be computed as follows.
\begin{eqnarray*}
  \gamma_n(-\rho_n + i \sqrt[s]{1 - \rho_n^2})^{\epsilon_n} & = & \gamma_n(-\rho_n + i \epsilon_n \sqrt[s]{1 - \rho_n^2}) \\
& = & \gamma_n (-\rho_n - \epsilon_n \rho_n \sqrt[+]{1 - \rho_n^{-2}}) \\
& = & -2 v_n - 2 v_n \epsilon_n \sqrt[+]{1 - \rho_n^{-2}}) \\
& \to & -2 v_n (1 + \epsilon_n).
\end{eqnarray*}
As to the third factor in (\ref{hatexpr}), observe that for $|\rho_n|$ large,
\begin{displaymath}
  |\phi(\rho_n)| = \Big| \int_{-1}^{\rho_n}\frac{dz}{\sqrt[s]{1 - z^2}} \Big| \leq \Big| \int_{-1}^1 \frac{dz}{\sqrt[s]{1 - z^2}} \Big| + \int_1^{|\rho_n|} \frac{dt}{\sqrt{t\!-\!1} \sqrt{t\!+\!1}} = O(\sqrt{|\rho_n|}).
\end{displaymath}
%\begin{eqnarray*}
%  |\phi(\rho_n)| & = & \Big| \int_{-1}^{\rho_n}\frac{dz}{\sqrt[s]{1 - z^2}} \Big| \\
%& \leq & c + \int_1^{|\rho_n|} \frac{dt}{\sqrt{t-1} \sqrt{t+1}} \\
%& = & O(\sqrt{|\rho_n|}).
%\end{eqnarray*}
Since $\hat{\zeta}_n = O(|\gamma_n|)$ by (\ref{zetagamma}), we thus conclude that $\phi(\rho_n) \hat{\zeta}_n \to 0$ and so
$$ e^{i \epsilon_n \phi(\rho_n) \hat{\zeta}_n} \to 1 \quad \textrm{as} \quad (b,a) \to (b_0, a_0). $$
Together with the result for $e^{i \eta_n^{(2)}}$ we conclude that for $(b,a) \to (b_0, a_0)$
$$ \gamma_n e^{i \eta_n} \to -2 v_n (1 + \epsilon_n) e^{\epsilon_n \chi_n(b_0,a_0)} $$
as claimed. The limit of $\gamma_n e^{-i \eta_n}$ is a simple variation of this argument.
\end{proof}

\begin{proof}[Proof of Proposition \ref{znpm}]
%With Lemma \ref{applemma} we now finish the proof of Proposition \ref{znpm}. 
We extend the functions $z_n^\pm$ to $D_n \cap \W$ as follows
\begin{equation} \label{zmpmformula}
z_n^\pm = \left\{ \begin{array}{ccc}
-2(1 \pm \epsilon_n) (\mu_n - \tau_n) e^{\pm \epsilon_n \chi_n} \; & \textrm{ on } &  D_n \cap X_n, \\
0 & \textrm{ on } & D_n \setminus X_n.
\end{array} \right.
\end{equation}

We have already seen that the functions $z_n^\pm$ are analytic on $\W \setminus D_n$. It is straightforward to verify that $z_n^\pm$ are continuous at every point of $D_n \cap X_n$ and of $D_n \setminus X_n$. Thus $z_n^\pm$ are continuous on all of $\W$. In view of Theorem A.6 in \cite{kapo} it remains to show that they are weakly analytic, when restricted to $D_n \cap \W$, i.e. that the restriction of $z_n^\pm$ to any one-dimensional complex disc $D$ contained in $D_n \cap \W$ is analytic. If the center of $D$ is in $X_n$, the entire disc $D$ is in $X_n$, if chosen sufficiently small. The analyticity of $z_n^\pm = \gamma_n e^{\pm i \eta_n}$ on $D$ is then evident from formula (\ref{zmpmformula}), the definition of $\chi_n$, and the local constancy of $\epsilon_n$ on $X_n$. If the center of $D$ does not belong to $X_n$ we argue as follows. The function $\mu_n - \tau_n$ is analytic on $D$. It either vanishes identically on $D$ in which case $z_n^\pm$ vanishes identically, too. Or it vanishes in only finitely many points. Outside these points, $D$ is in $X_n$, hence $z_n^\pm$ is analytic there. By continuity and analytic continuation, these functions are analytic on all of $D$. We thus have shown that $z_n^\pm$ are analytic on $D$. As $\chi_n$ is analytic and $\epsilon_n$ is locally constant, it follows that $z_n^\pm$ is weakly analytic on $D_n \cap \W$. This proves the analyticity of $z_n^\pm$ on $\W$.
\end{proof}

\section{Proof of Proposition \ref{dbaznnew}} \label{proofdbazn}

To prove Proposition \ref{dbaznnew} we follow the arguments used in \cite{kapo} to show similar results for KdV. We begin with some preparations for the proof of Proposition \ref{dbaznnew}. To compute the gradient of $z_n^+$ at a point $(b,a)$ in $\M \cap D_n$, we approximate $(b,a)$ by elements $(b', a')$ in
\begin{equation} \label{bndef}
  B_n := \{ (b,a) \in \M \setminus D_n: \mu_n = \tau_n \textrm{ and sign} \sqrt[*]{\Delta^2_{\mu_n} - 4} = (-1)^{N+n+1} \}.
\end{equation}
It follows from the results of the spectral theory of Jacobi matrices reviewed in section \ref{tools} that $B_n \neq \emptyset$. As a preliminary step towards the computation of $\nba z_n^+$, we need the following two lemmas.

\begin{lemma} \label{zntnmnlemma}
For any $(b,a) \in \M \cap D_n$,
\begin{equation} \label{zntnmn}
  \nba z_n^+ = 2 (\nba \tau_n - \nba \mu_n) + i \lim_{B_n \ni (b', a') \to (b,a)} (f_{2n+1}^\mbf2 - f_{2n}^\mbf2),
\end{equation}
where for $i \in \{ 2n,2n+1 \}$ and $(b',a') \in B_n$ we denote by $f_i$ the eigenvector of $L(b', a')$ associated to $\l_i$, normalized by
\begin{displaymath}
\sum_{j=1}^N f_i(j)^2 = 1 \quad \textrm{and} \quad \big( f_i(1), f_i(2) \big) \in (\R_{>0} \times \R) \cup (\{ 0 \} \times \R_{>0}).
\end{displaymath}
\end{lemma}

\begin{proof}[Proof of Lemma \ref{zntnmnlemma}]
Recall that we have introduced $\psi_n(\l)$ and $\zeta_n(\l)$ in (\ref{psiprodrepr}) and (\ref{zetafndef}), respectively. 
%We perform a computation very similar to the one used for the proof of Lemma \ref{applemma}: 
For $(b', a') \in B_n$,
\begin{displaymath}
  \textrm{sign} \; \psi_n(\mu_n) = (-1)^{N+n+1} \quad \textrm{and} \quad \textrm{sign} \; \zeta_n(\mu_n) = 1
\end{displaymath}
% sign $\psi_n(\mu_n) = (-1)^{N+n+1}$ and sign $\zeta_n(\mu_n) = 1$ 
(see discussions after (\ref{psiprodrepr}) and (\ref{psiexp})), and hence the identity (\ref{psiexp}) reads
\begin{displaymath}
  \frac{\psi_n(\mu_n)}{\sqrt[*]{\Delta^2_{\mu_n} - 4}} = \frac{\zeta_n(\mu_n)}{\sqrt[+]{(\l_{2n+1} - \mu_n)(\mu_n - \l_{2n})}}.
\end{displaymath}
Going through the calculations in the proof of Lemma \ref{applemma} with (\ref{defeps}) replaced by the latter identity, all $s$-roots replaced by the principal branch $\sqrt[+]\cdot$, and with $\epsilon_n = 1$, we can write
\begin{displaymath}
z_n^+ = \gamma_n e^{i \eta_n} = \gamma_n e^{i \eta_n^{(1)}} \cdot e^{i \eta_n^{(2)}}
\end{displaymath}
where, mod $2\pi$,
\begin{displaymath}
  \eta_n^{(1)} = \zeta_n(\l_{2n}) \int_{\l_{2n}}^{\mu_n} \frac{d\l}{\sqrt[+]{(\l_{2n+1} - \l)(\l - \l_{2n})}}
\end{displaymath}
and
\begin{displaymath}
\eta_n^{(2)} = \int_{\l_{2n}}^{\mu_n} \frac{\zeta_n(\l) - \zeta_n(\l_{2n})}{\sqrt[+]{(\l_{2n+1} - \l)(\l - \l_{2n})}} \, d\l.
\end{displaymath}
Note that with $\l = \tau_n + z \gamma_n/2$, $\sqrt[+]{(\l_{2n+1} - \l)(\l - \l_{2n})} = \frac{\gamma_n}{2}\sqrt{1-z^2}$ and hence on $B_n$
\begin{displaymath}
  \eta_n^{(1)} = \zeta_n(\l_{2n})  \int_{-1}^0 \frac{dz}{\sqrt{1-z^2}} = \frac{\pi}{4} \zeta_n(\l_{2n})
\end{displaymath}
In view of (\ref{zetagamma}) we then get in the limit $(b',a') \to (b,a)$, with $(b',a') \in B_n$,
\begin{displaymath}
  \eta_n^{(1)} \to \frac{\pi}{4} \; (\textrm{mod} \, 2\pi) %\int_{\l_{2n}}^{\mu_n} \frac{\zeta_n(\l_{2n})}{\sqrt[+]{(\l_{2n+1} - \l)(\l - \l_{2n})}} \, d\l \quad \textrm{mod} \; 2\pi
\end{displaymath}
Using again $\l \equiv \l(z) = \tau_n + z \gamma_n/2$ one computes for $(b',a') \in B_n$
\begin{displaymath}
  \eta_n^{(2)} = \int_{-1}^0 \left( \int_0^1 \zeta_n'(\l_{2n} + s(\l - \l_{2n})) ds \right) \frac{\gamma_n (1+z)}{2 \sqrt{1-z^2}} dz
\end{displaymath}
and thus $\eta_n^{(2)} \to 0$ as $(b',a') \to (b,a)$, or $e^{i \eta_n''} \to 1$.

Since for $(b,a) \in \M \cap D_n$, one has $\gamma_n e^{i \eta_n'} = 0$ by Proposition \ref{znpm}, it then follows that
\begin{displaymath}
  \nba z_n^+ = \lim_{B_n \ni (b', a') \to (b,a)} \nba \big( \gamma_n e^{i \eta_n'} \big),
\end{displaymath}
%Moreover, with $-1 \leq \rho_n \leq 1$ and $v_n$ defined by $v_n := \mu_n - \tau_n =: \rho_n \gamma_n / 2$, and with $\hat{\zeta}_n := \zeta_n(\l_{2n}) - 1$, we get for $(b',a') \in B_n$,
Moreover, let $v_n := \mu_n - \tau_n$ and as in (\ref{rhodefine}), introduce $\rho_n = \frac{\mu_n - \tau_n}{\gamma_n/2}$. Then $-1 \leq \rho_n \leq 1$ and by (\ref{phiexprlambda}) we get for $(b',a') \in B_n$, $e^{i \eta_n^{(1)}} = e^{i \phi(\rho_n)} e^{i \phi(\rho_n) \hat{\zeta}_n}$, where $\hat{\zeta}_n = \zeta_n(\l_{2n}) - 1$. By (\ref{expphiformula}) it follows that
\begin{displaymath}
  \gamma_n e^{i \eta_n^{(1)}} = \left( -2 v_n + i \gamma_n \sqrt[+]{1 - \rho_n^2} \, \right) \cdot \left( -\rho_n + i \sqrt[+]{1 - \rho_n^2} \, \right)^{\hat{\zeta}_n}.
\end{displaymath}
%\begin{eqnarray*}
%  \gamma_n e^{i \eta_n^{'}} & = & \gamma_n \big( -\rho_n + i \sqrt[+]{1 - \rho_n^2} \, \big)^{\zeta_n(\l_{2n})} \\
%& = & \big( -2 v_n + i \gamma_n \sqrt[+]{1 - \rho_n^2} \, \big) \cdot \big( -\rho_n + i \sqrt[+]{1 - \rho_n^2} \, \big)^{\hat{\zeta}_n}.
%\end{eqnarray*}
The gradients of both factors have a limit as $(b',a') \to (b,a)$, and the product rule can be applied. For $(b',a') \in B_n$, we have $\mu_n = \tau_n$ and hence $v_n = 0$ as well as $\rho_n = 0$. Thus, the first factor equals $i \gamma_n$ and hence, in the limit, vanishes, while the second factor equals $(i)^{\hat{\zeta}_n}$ and thus, by (\ref{zetagamma}), converges to $1$. As a consequence, by the product rule,
{\setlength\arraycolsep{2pt}
\begin{eqnarray*}
  \nba z_n^+ & = & \lim_{(b', a') \to (b,a)} \nba \big( \gamma_n e^{i \eta_n'} \big) \\
& = & \lim_{(b', a') \to (b,a)} \nba \big( -2 v_n + i \gamma_n \sqrt[+]{1 - \rho_n^2} \big) \\
& = & 2 (\nba \tau_n - \nba \mu_n) + i \lim_{(b', a') \to (b,a)} \nba \gamma_n \\
& = & 2 (\nba \tau_n - \nba \mu_n) + i \lim_{(b', a') \to (b,a)} \big( f_{2n+1}^\mbf2 - f_{2n}^\mbf2 \big)
\end{eqnarray*}}
where for the latter identity we used that on $\M \setminus D_n$, $\nba \l_i = f_i^\mbf2$ for $i \in \{ 2n, 2n+1 \}$ (cf. \cite{ahtk1}, Proposition 5.3) and that $\lim_{(b', a') \to (b,a)} \big( f_{2n+1}^\mbf2 - f_{2n}^\mbf2 \big)$ exists, as $z_n^+$, $\mu_n$, and $\tau_n$ are analytic.
\end{proof}

\begin{lemma} \label{bmamlimit}
As $(b', a') \in B_n$ tends to $(b,a) \in \M \cap D_n$, the periodic eigenvectors $f_{2n}$ and $f_{2n+1}$ of $L(b',a')$, normalized as in Lemma \ref{zntnmnlemma}, converge to normalized eigenvectors of $L(b,a)$, denoted by the same symbols, such that in the limit,
\begin{displaymath}
\langle f_{2n+1}, g_n \rangle < 0 < \langle f_{2n}, g_n \rangle
\end{displaymath}
and
\begin{equation} \label{flimeqn}
  f_{2n}(1) \, \langle f_{2n+1}, g_n \rangle = - f_{2n+1}(1) \, \langle f_{2n}, g_n \rangle.
\end{equation}
\end{lemma}

\begin{proof}[Proof of Lemma \ref{bmamlimit}]
As $(b',a')$ tends to $(b,a)$, the initial data $\big( f_i(1), f_i(2) \big)$ of the normalized eigenvectors $f_i$ ($i \in \{ 2n, 2n+1 \}$) is a vector in the unit disc of $\R^2$. Choose a convergent subsequence of initial data. Then, for $i \in \{ 2n, 2n+1 \}$, $f_i$ can be expressed as a linear combination of the fundamental solutions $y_1$ and $y_2$, $f_i = f_i(0) y_1 + f_i(1) y_2$. As the fundamental solutions depend analytically on $\l$ and $(b',a')$, the eigenvectors $f_{2n+1}$ and $f_{2n}$ then converge to some eigenvectors of $L(b,a)$, which we denote by $\bar{f}_{2n+1}$ and $\bar{f}_{2n}$. Note that, by the normalization of $f_i$, one has $\bar{f}_i(1) \geq 0$ for $i \in \{ 2n, 2n+1 \}$.

By (\ref{zntnmn}), $\lim_{(b', a') \to (b,a)} \left( f_{2n+1}^\mbf2 - f_{2n}^\mbf2 \right)$ exists. As $\tau_n$ is analytic and, on $\M\setminus D_n$, $\nba \tau_n = \left( f_{2n+1}^\mbf2 + f_{2n}^\mbf2 \right) / 2$, it follows that $\lim_{(b', a') \to (b,a)} \left( f_{2n+1}^\mbf2 + f_{2n}^\mbf2 \right)$ exists as well. Hence the limits of $f_i^\mbf2$, $i \in \{ 2n, 2n+1 \}$ exist, and $\bar{f}_{2n+1}$ and $\bar{f}_{2n}$ are uniquely determined up to a sign. As $\bar{f}_i(1) \geq 0$ for $i \in \{ 2n, 2n+1 \}$, this sign is uniquely determined once we show that
\begin{equation} \label{fneq0}
  \bar{f}_i(1) \neq 0.
\end{equation}

To simplify notation, write temporarily $f$ and $g$ instead of $f_{2n+1}$ and $g_n$. To prove (\ref{fneq0}), observe that
\begin{eqnarray*}
  (\l_{2n+1} - \mu_n) \langle f, g \rangle & = & \langle \l_{2n+1} f, g \rangle - \langle f, \mu_n g \rangle = \langle L f, g \rangle - \langle f, L g \rangle \\
& = & \sum_{j=1}^N \Big( \left[ b_j f(j) + a_j f(j+1) + a_{j-1} f(j-1) \right] g(j) \\
&& \qquad - f(j) \left[ b_j g(j) + a_j g(j+1) + a_{j-1} g(j-1) \right] \Big) \\
& = & \sum_{j=1}^N \big( a_j f(j+1) g(j) - a_{j-1} f(j) g(j-1) \big) \\
&& + \sum_{j=1}^N \big( a_{j-1} f(j-1) g(j) - a_j f(j) g(j+1) \big) .
\end{eqnarray*}
Note that the latter two sums are telescoping, hence
\begin{eqnarray*}
  (\l_{2n+1} - \mu_n) \langle f, g \rangle & = & \, a_N \big( f(N\!\!+\!\!1) g(N) - f(1) g(0) + f(0) g(1) - f(N) g(N\!\!+\!\!1) \big) \\
& = & \, a_N f(1) \big( (-1)^{n+N} g(N) - g(0) \big),
\end{eqnarray*}
where for the latter equality we used that $g(1) = 0 = g(N+1)$ and $f(N+1) = (-1)^{n+N} f(1)$ according to whether $\l_{2n+1}$ is a periodic or antiperiodic eigenvalue - see (\ref{deltalambdapm2}) in section \ref{tools}. Hence we have
\begin{equation} \label{l2n1mu}
  (\l_{2n+1} - \mu_n) \langle f_{2n+1}, g_n \rangle = a_N f_{2n+1}(1) \big( (-1)^{n+N} g_n(N) - g_n(0) \big).
\end{equation}
A similar computation shows that
\begin{equation} \label{l2nmu}
  (\l_{2n} - \mu_n) \langle f_{2n}, g_n \rangle = a_N f_{2n}(1) \big( (-1)^{n+N} g_n(N) - g_n(0) \big).
\end{equation}
%The sign $(-1)^{n+N}$ in (\ref{l2n1mu}) and (\ref{l2nmu}) follows from the ordering of periodic and antiperiodic eigenvalues in the spectrum of $Q$, as described in section \ref{tools}. Taking the quotient of (\ref{l2n1mu}) and (\ref{l2nmu}), observing that by construction $\l_{2n+1} - \mu_n = \mu_n - \l_{2n}$ 
For $(b', a') \in B_n$, one has $\l_{2n+1} - \mu_n = \mu_n - \l_{2n}$ as well as $f_i(1) > 0$ ($i \in \{ 2n, 2n+1 \}$) and $(-1)^{n+N} g_n(N) \neq g_n(0)$. Hence the quotients of the left and right hand sides of (\ref{l2n1mu}) and (\ref{l2nmu}) are well defined, and we obtain $f_{2n}(1) \langle f_{2n+1}, g_n \rangle = - f_{2n+1}(1) \langle f_{2n}, g_n \rangle$. Passing to the limit as $(b',a') \to (b,a)$ we obtain
\begin{equation} \label{intermedbareqn}
  \bar{f}_{2n}(1) \langle \bar{f}_{2n+1}, g_n \rangle = - \bar{f}_{2n+1}(1) \langle \bar{f}_{2n}, g_n \rangle.
\end{equation}

We claim that
\begin{equation} \label{limpos}
  \kappa_n := \lim_{(b', a') \to (b,a)} \frac{(-1)^{n+N} g_n(N) - g_n(0)}{\l_{2n+1} - \mu_n}
\end{equation}
exists and that $\kappa_n < 0$. To see it, divide (\ref{l2n1mu}) by $(\l_{2n+1} - \mu_n)$. Then the existence of the limit in (\ref{limpos}) implies that one can take limits of both sides of the resulting equation as $(b', a') \to (b,a)$ to get
\begin{equation} \label{f2n+1gnkappan}
  \langle \bar{f}_{2n+1}, g_n \rangle = a_N \kappa_n \bar{f}_{2n+1}(1).
\end{equation}
If $\bar{f}_{2n+1}(1) = 0$, then, as $\bar{f}_{2n+1}$ is periodic or antiperiodic, $\bar{f}_{2n+1}(N+1) = 0$ as well. Hence $\bar{f}_{2n+1}$ is a (nontrivial) scalar multiple of $g_n$ and thus $\langle \bar{f}_{2n+1}, g_n \rangle \neq 0$, contradicting (\ref{f2n+1gnkappan}). %This shows that (\ref{limpos}) 
Hence the claim that $\kappa_n < 0$ implies that $\bar{f}_{2n+1}(1) > 0$, i.e. $\bar{f}_{2n+1}$ satisfies all the normalization conditions listed in Lemma \ref{bmamlimit}.

It remains to prove that the limit (\ref{limpos}) exists and that $\kappa_n < 0$, or equivalently,
\begin{equation} \label{limposequiv}
   \lim_{(b', a') \to (b,a)} \frac{(-1)^{n+N} y_1(N, \mu_n) - 1}{\l_{2n+1} - \mu_n} < 0.
\end{equation}
Recall that $\sqrt[*]{\Delta^2(\mu_n) - 4} = y_1(N, \mu_n) - y_2(N+1, \mu_n)$. Hence, for $(b', a') \in B_n$,
\setlength\arraycolsep{2pt} {\begin{eqnarray*}
  2 y_1(N, \mu_n) & = & \big( y_1(N, \mu_n) + y_2(N+1, \mu_n) \big) + \big( y_1(N, \mu_n) - y_2(N+1, \mu_n) \big) \\
& = & \Delta(\mu_n) + \sqrt[*]{\Delta^2(\mu_n) - 4} \\
& = & \Delta(\mu_n) + (-1)^{N-n-1} \sqrt[+]{\Delta^2(\mu_n) - 4},
\end{eqnarray*}}
the last equality being a consequence of the definition (\ref{bndef}) of $B_n$. Recall that, according to (\ref{deltalambdapm2}), $2 = (-1)^{n+N} \Delta(\l_{2n+1})$. Substituting $4 = \Delta^2(\l_{2n+1})$ into the formula above, the inequality (\ref{limposequiv}) can then be equivalently written as
\begin{equation} \label{limpos3}
  \lim_{(b', a') \atop \to (b,a)} \left( (-1)^{n+N} \frac{\Delta(\mu_n) - \Delta(\l_{2n+1})}{\l_{2n+1} - \mu_n} - \frac{\sqrt[+]{\Delta^2(\mu_n) - \Delta^2(\l_{2n+1})}}{\l_{2n+1} - \mu_n} \right) < 0.
\end{equation}
Concerning the first term in the above expression, we get in the limit, as $(b', a') \to (b,a)$
\begin{displaymath}
  \frac{\Delta(\mu_n) - \Delta(\l_{2n+1})}{\mu_n - \l_{2n+1}} \to \dot{\Delta}(\l_{2n+1}) = 0,
\end{displaymath}
as $\l_{2n+1}$ is a double eigenvalue of $Q(b,a)$. Concerning the second term in (\ref{limpos3}), write
\begin{eqnarray}
  \Delta^2(\mu_n) - \Delta^2(\l_{2n+1}) & = & -2 \int_{\mu_n}^{\l_{2n+1}} \Delta(\l) \dot{\Delta}(\l) d\l \nonumber\\
& = & -2 \int_{\mu_n}^{\l_{2n+1}} \Delta(\l) \left( \int_{\dot{\l}_n}^\l \ddot{\Delta}(\mu) \, d\mu \right) d\l, \label{doubleint}
\end{eqnarray}
where $\dot{\l}_n$ is the unique root of $\dot{\Delta}$ in the $n$-th gap. Note that locally uniformly on $\M$,
\begin{displaymath}
  \int_{\dot{\l}_n}^\l \ddot{\Delta}(\mu) \, d\mu = (\l - \dot{\l}_n) \ddot{\Delta}(\dot{\l}_n) + O((\l - \dot{\l}_n)^2)%\int_{\dot{\l}_n}^\l \big( \ddot{\Delta}(\mu) - \ddot{\Delta}(\dot{\l}_n) \big) d\mu
\end{displaymath}
and
\begin{displaymath}
    -2 \int_{\mu_n}^{\l_{2n+1}} \!\!\! \Delta(\l) (\l \!\!-\!\! \dot{\l}_n) \ddot{\Delta}(\dot{\l}_n) d\l = -2 \, \Delta(\mu_n) \ddot{\Delta}(\dot{\l}_n) \int_{\mu_n}^{\l_{2n+1}} (\l - \dot{\l}_n) d\l + O(\gamma_n^3)
\end{displaymath}
The first term on the right hand side of the latter expression can be computed to be, using that $\mu_n = \tau_n$ on $B_n$ and $\dot{\l}_n = \tau_n + O(\gamma_n^2)$ locally uniformly
\begin{displaymath}
  -2 \, \Delta(\mu_n) \ddot{\Delta}(\dot{\l}_n) \frac{1}{2} \left( (\l_{2n+1} \!-\! \dot{\l}_n)^2 \!-\! (\mu_n \!-\! \dot{\l}_n)^2 \right) = -\Delta(\mu_n) \ddot{\Delta}(\dot{\l}_n) \left( \frac{\gamma_n}{2} \right)^2 +  O(\gamma_n^3).
\end{displaymath}
%Again using that $\dot{\l}_n = \tau_n + O(\gamma_n^2)$ locally uniformly and $\mu_n = \tau_n$ on $B_n$, it then follows that
Dividing (\ref{doubleint}) by $(\mu_n - \l_{2n+1})^2$ and taking the limit thus leads to
\begin{displaymath}
  \lim_{(b', a') \to (b,a)} \frac{\Delta^2(\mu_n) - \Delta^2(\l_{2n+1})}{(\mu_n - \l_{2n+1})^2} = - \Delta(\l_{2n+1}) \ddot{\Delta}(\l_{2n+1}).
\end{displaymath}
As $\Delta(\l_{2n+1}) = (-1)^{n+N} \cdot 2$ (see (\ref{deltalambdapm2})) and $(-1)^{n+N} \ddot{\Delta}(\l_{2n+1}) < 0$ (again by (\ref{deltalambdapm2}) and the fact that $\Delta(\l)$ is a polynomial of degree $N$), one concludes that $-\Delta(\l_{2n+1}) \ddot{\Delta}(\l_{2n+1}) > 0$. This proves the estimate (\ref{limpos3}), hence by (\ref{l2n1mu}), $\langle \bar{f}_{2n+1}, g_n \rangle < 0$. By a similar argument, one shows that $\bar{f}_{2n}(1) > 0$, and therefore, (\ref{intermedbareqn}) implies $\langle f_{2n}, g_n \rangle > 0$.
\end{proof}

\begin{proof}[Proof of Proposition \ref{dbaznnew}]
According to Lemma \ref{bmamlimit},
\begin{equation} \label{limf2n+1f2n}
  \lim_{(b', a') \to (b,a)} \big( f_{2n+1}^\mbf2 - f_{2n}^\mbf2 \big) = f_{2n+1}^\mbf2 - f_{2n}^\mbf2
\end{equation}
where the limiting eigenvectors $f_{2n}$ and $f_{2n+1}$ are orthonormal and satisfy the inequalities $\langle f_{2n+1}, g_n \rangle < 0 < \langle f_{2n}, g_n \rangle$. By definition, $g_n$ and $h_n$ are orthonormal and span the same subspace as $f_{2n}$ and $f_{2n+1}$. Hence there exist $s,t \geq 0$ with $s^2 + t^2 = 1$ such that
\begin{displaymath} %\label{fhngn}
  \begin{array}{rcl}
f_{2n+1} & \, = \, & s \, h_n - t g_n, \\
f_{2n} & \, = \, & t \, h_n + s g_n.
  \end{array}
\end{displaymath}
Substituting these formulas into equation (\ref{flimeqn}), we obtain $t^2 = s^2$, and hence $s=t=\frac{1}{\sqrt{2}}$. Thus we have
\begin{equation}
\begin{array}{ccc}
f_{2n+1}^\mbf2 - f_{2n}^\mbf2 & = & -2 h_n \cds g_n, \label{f2n+1-f2n} \\
f_{2n+1}^\mbf2 + f_{2n}^\mbf2 & = & h_n^\mbf2 + g_n^\mbf2. \nonumber
\end{array}
\end{equation}
By \cite{ahtk1}, $\nba \mu_n = g_n^\mbf2$ and $\nba \tau_n = (f_{2n+1}^\mbf2 + f_{2n}^\mbf2)/2$, hence
\begin{equation} \label{nbataunnbamun}
  2 (\nba \tau_n - \nba \mu_n) = h_n^\mbf2 - g_n^\mbf2.
\end{equation}
In view of (\ref{limf2n+1f2n})-(\ref{nbataunnbamun}), formula (\ref{dbaznformulaapp}) then follows from (\ref{zntnmn}).
\end{proof}

\section{Proof of Theorem \ref{sumthm}} \label{diffeochapter}

In this section we show Theorem \ref{sumthm}. Its three statements are contained in Theorem \ref{diffeo}, Theorem \ref{candiffeo}, and Corollary \ref{hamfunctact}, respectively.

\begin{theorem} \label{diffeo}
The map
\begin{displaymath}
\begin{array}{ccll}
 \Omega: & \M & \to & \P \\
 & (b,a) & \mapsto & ((x_n, y_n)_{1 \leq n \leq N-1}, C_1, C_2)
\end{array}
\end{displaymath}
is a global, real analytic diffeomorphism.
\end{theorem}

\emph{Local Properties}
In a first step we establish that $\Omega$ is a local diffeomorphism everywhere in phase space.

\begin{prop} \label{local}
At every point $(b,a) \in \M$, the differential $d_{(b,a)}\Omega: T_{(b,a)}\M \to T_{\Omega(b,a)}\P$ is a linear isomorphism.
\end{prop}
%In particular, on the dense open subset $\M \setminus \cup_{n=1}^{N-1} D_n$, i.e. on the set where all gaps are open, the $N-1$ actions $(I_n)_{1 \leq n \leq N-1}$ are $N-1$ integrals in involution in the sense of Liouville integrabls systems.

Let us first introduce some additional notation. For $(b,a) \in \M$ and $1 \leq n \leq N-1$, define
$$ d_n := \nba x_n \quad \textrm{and} \quad d_{-n} := \nba y_n. $$

Further we recall Lemma 7.2 in \cite{ahtk1}, needed later.
\begin{lemma} \label{actindep}
At every point $(b,a)$ in $\M$, the vectors
\begin{displaymath}
\big( (\nba I_n)_{n \in K}, \nba C_1, \nba C_2 \big)
\end{displaymath}
are linearly independent. Here $K = K(b,a)$ denotes the index set of open gaps,
$$ K := \{ 1 \leq n \leq N-1: \gamma_n(b,a) > 0 \}. $$
\end{lemma}

Proposition \ref{local} follows from the following lemma.
\begin{lemma} \label{gradfamilylemma}
At every point $(b,a) \in \M$, the $2N$ vectors
\begin{equation} \label{gradfamily}
\big( (d_n)_{1 \leq n \leq N-1}, \, (d_{-n})_{1 \leq n \leq N-1}, \, \nba C_1, \, \nba C_2 \big)
\end{equation}
are linearly independent.
\end{lemma}

\begin{proof}[Proof of Lemma \ref{gradfamilylemma}]
To verify the claimed statement, consider an arbitrary linear combination $f = \sum_{1 \leq |n| \leq N-1} r_n d_n + s_1 \nba C_1 + s_2 \nba C_2$ with real coefficients $(r_n)_{1 \leq |n| \leq N}$, $s_1$, $s_2$ such that $f = 0$. For $m \in K$ take the scalar product of both sides of $f = 0$ with $J \nba I_m$ and use Lemma \ref{canrel} together with the identity $I_m = (x_m^2 + y_m^2)/2$ to get
$$ 0 = \langle f, J \nba I_m \rangle = r_m \{ x_m, I_m \}_J + r_{-m} \{ y_m, I_m \}_J = r_m y_m - r_{-m} x_m. $$
Hence the $2$-vectors $(r_m, r_{-m})$ and $(x_m, y_m)$ are parallel, i.e. $(r_m, r_{-m}) = c_m (\cos \theta_m, \sin \theta_m)$ with $c_m \in \R$ satisfying $c_m^2 = r_m^2 + r_{-m}^2$. Thus, if $c_m$ vanishes, $r_m$ and $r_{-m}$ both vanish. Furthermore, by the definition of $d_{\pm m}$,
\begin{eqnarray*}
r_m d_m + r_{-m} d_{-m} & = & r_m \nba x_m + r_{-m} \nba y_m \\
& = & c_m (\cos \theta_m \nba x_m + \sin \theta_m \nba y_m) \\
& = & \frac{c_m}{\sqrt{2I_m}} \nba I_m,
\end{eqnarray*}
where for the last equality we used the identity
$$ \nba I_m = \sqrt{2I_m}(\cos \theta_m \nba x_m + \sin \theta_m \nba y_m), $$
obtained from $2 I_m = x_m^2 + y_m^2$ by differentiation.

Hence the equation $f = 0$ reads
\begin{equation} \label{mixed}
\sum_{n \in K} \frac{c_n}{\sqrt{2I_n}} \nba I_n + \sum_{n \notin K} (r_n \nba x_n + r_{-n} \nba y_n) + s_1 \nba C_1 + s_2 \nba C_2 = 0.
\end{equation}
Next, for $m \notin K$, take the scalar product of both sides of (\ref{mixed}) with $J \nba y_m$ and $J \nba x_m$. By the commutator relations of Lemma \ref{canrel} one then obtains the identities
\begin{displaymath}
0 = \langle f, J \nba y_m \rangle = r_m \{ x_m, y_m \}_J = r_m
\end{displaymath}
and
\begin{displaymath}
0 = \langle f, J \nba x_m \rangle = -r_{-m} \{ x_m, y_m \}_J = -r_{-m}.
\end{displaymath}
Hence (\ref{mixed}) becomes
$$ \sum_{n \in K} \frac{c_n}{\sqrt{2I_n}} \nba I_n + s_1 \nba C_1 + s_2 \nba C_2 = 0. $$
By Lemma \ref{actindep}, $c_n = 0$ - hence $r_n = r_{-n} = 0$ by the remark above - for any $n \in K$ and $s_1 = s_2 = 0$.
\end{proof}

\emph{Global Properties} In a second step, we show that $\Omega$ is bijective and canonical. First we show

\begin{prop} \label{global}
The map $\Omega$ is proper, i.e. the preimage of any compact set is compact.
\end{prop}

To prove Proposition \ref{global} we need two auxiliary results.

\begin{lemma} \label{est1}
For any $(b,a) \in \M$ and any $1 \leq n \leq N-1$,
\begin{equation} \label{estingnorig}
\gamma_n^2 \leq 3\pi (\l_{2N} - \l_1)I_n
\end{equation}
and
\begin{equation} \label{estingnsumorig}
  \sum_{n=1}^{N-1} \gamma_n^2 \leq 12 \pi^2 \alpha \left( \sum_{n=1}^{N-1} I_n \right) + 9 \pi^2 (N-1) \left( \sum_{n=1}^{N-1} I_n \right)^2.
\end{equation}
\end{lemma}
A proof of Lemma \ref{est1} can be found in Appendix A of \cite{ahtk1}; in the case $\a = 1$ (\ref{estingnorig}) has been proved in \cite{bbgk} (p.601-602).

\begin{lemma} \label{est2}
For $1 \leq n \leq N$ and any $(b,a) \in \M$,
\setlength\arraycolsep{2pt} {\begin{eqnarray}
|b_n| & \leq & |C_1(b,a)| + \sum_{k=1}^{N-1} \gamma_k(b,a), \label{estb} \\
0 < a_n & \leq & N \left( 2 \pi C_2(b,a) + |C_1(b,a)| + \sum_{k=1}^{N-1} \gamma_k(b,a) \right), \label{esta} \\
\l_{2N}(b,a) - \l_1(b,a) & \leq & 2 \pi C_2(b,a) + \sum_{k=1}^{N-1} \gamma_k(b,a). \label{estl}
\end{eqnarray}}
%where $C_1 = -\frac{1}{N} \sum_{n=1}^N b_n$ and $C_2 = (\prod_{n=1}^N a_n)^\frac{1}{N}$.
\end{lemma}
Lemma \ref{est2} is shown in Appendix \ref{proofest}. In a weaker form, it has been proved in \cite{bggk} (p.564-565).

\begin{proof}[Proof of Proposition \ref{global}]
Let $(b^{(m)}, a^{(m)})_{m \geq 1} \subseteq \M$ be a sequence in $\M$ so that $\big(\Omega (b^{(m)}, a^{(m)})\big)_{m \geq 1}$ converges in $\P$. Then, for any $1 \leq n \leq N-1$, the sequence $\left( I_n(b^{(m)}, a^{(m)}) \right)_{m \geq 1}$ of action variables is a Cauchy sequence, as well as the sequence $(C_1^{(m)}, C_2^{(m)})_{m \geq 1}$ of the values of the Casimir functions $C_1$ and $C_2$. By Lemma \ref{est1} and Lemma \ref{est2} it then follows that $(b^{(m)}, a^{(m)})_{m \geq 1}$ admits a subsequence which converges to an element $(b,a)$ in $\R^N \times \R_{\geq 0}^N$. As by assumption, the sequence $\big( C_2(a^{(m)}) \big)_{m \geq 1}$ converges to $\a > 0$ and
\begin{displaymath}
  \a = \lim_{m \to \infty} C_2(a^{(m)}) = \left( \prod_{n=1}^N a_n \right)^{1/N}
\end{displaymath}
it follows that $a_n > 0$ for all $1 \leq n \leq N$, i.e. $(b,a) \in \M$.
\end{proof}

\begin{proof}[Proof of Theorem \ref{diffeo}]
By Proposition \ref{local}, $\Omega$ is open, and by Proposition \ref{global}, it is closed. As $\P$ is connected, $\Omega(\M) = \P$, hence $\Omega$ is onto. It is also 1-1, since, by the same reasoning, the set $\mathcal{B}$ of all points in $\P$ with more than one preimage is open and closed. We claim that $(0_{N-1}, 0_{N-1}, 0, 1) \notin \mathcal{B}$ and hence $\mathcal{B} = \emptyset$. Here $0_{N-1}$ denotes the vector $(0, \ldots, 0) \in \R^{N-1}$. To see that $(0_{N-1}, 0_{N-1}, 0, 1) \notin \mathcal{B}$ note that for any $(b,a) \in \M$ with $\Omega(b,a) = (0_{N-1}, 0_{N-1}, 0, 1)$, all action variables vanish, and hence, by Theorem \ref{analytic} (ii), all gaps must be collapsed. By the results reviewed in section \ref{tools}, it then follows that $\mu_n = \l_{2n}$ for any $1 \leq n \leq N-1$ and hence there is exactly one matrix $Q$ with $\gamma_n = 0$ for any $1 \leq n \leq N-1$ and $(\b, \a) = (0,1)$. Since we have already shown the real analyticity of $\Omega$ in Theorem \ref{coordanaly}, this completes the proof of Theorem \ref{diffeo}.
\end{proof}

Next we show that $\Omega: \M \to \P$ is canonical. Recall that the phase space $\M$ is endowed with the Poisson bracket $\{ \cdot, \cdot \}_J$ defined by (\ref{poisson}). It is degenerate and has $C_1$, $C_2$ as Casimir functions. The model space $\P$ is endowed with the standard Poisson structure on $\R^{2(N-1)}$, i.e. among the coordinate functions $(x_n, y_n)_{1 \leq n \leq N-1}, \b, \a$, all Poisson brackets vanish, except for $1 \leq n \leq N-1$,
$$ \{ x_n, y_n \}_0 = -\{ y_n, x_n \}_0 = 1. $$
Note that on the model space $\P$, the coordinate functions $\b$ and $\a$ are two independent Casimirs defining a trivial foliation with symplectic leaves $\Pba := \R^{2(N-1)} \times \{ \b \} \times \{ \a \}$.

\begin{theorem} \label{candiffeo}
The map $\Omega$ is canonical, i.e. it preserves the Poisson brackets.
\end{theorem}

\begin{proof}
We have to verify that $\{ F, G \}_0 \circ \Omega = \{ F \circ \Omega, G \circ \Omega \}_J$ for arbitrary functions $F, G$ in $C^1(\P)$. Clearly, the pullbacks of the functions $\b, \a$ on $\P$ are the Casimir functions $C_1, C_2$ of $\{ \cdot, \cdot \}_J$. Moreover, by Lemma \ref{canrel}, $\{ x_k, y_l \}_J = \delta_{kl}$ and $\{ x_k, x_l \}_J = \{ y_k, y_l \}_J = 0$ everywhere on $\M$, hence $\Omega$ is canonical.
\end{proof}

Let $\pi: \P \to \R \times \R_{>0}$ denote the projection of $\P = \R^{2(N-1)} \times \R \times \R_{>0}$ onto the last two factors. Then $\pi$ defines a symplectic foliation with leaves $\Pba = \R^{2(N-1)} \times \{ \b \} \times \{ \a \}$. The definition of $\Omega$ together with Theorem \ref{candiffeo} then leads to the following result.

\begin{cor} \label{respfol}
For every $\b \in \R$ and $\a > 0$,
$$ \Omega(\Mba) = \Pba, $$
and $\Omega|_{\Mba}: \Mba \to \Pba$ is a symplectomorphism. In particular, the map
\begin{displaymath}
  (I_n)_{1 \leq n \leq N-1}: \Mba \to \big( \R_{\geq 0} \big)^{N-1}
\end{displaymath}
is onto.
\end{cor}

To formulate the last result of this section, recall that in section \ref{tools}, for any $(b,a) \in \M$, we have introduced the isospectral set
\begin{displaymath}
\textrm{Iso} \, (b,a) = \{ (b',a') \in \M: \textrm{spec} \; Q(b',a') = \textrm{spec}\; Q(b,a) \}.
\end{displaymath}
For $(x, y, \b, \a) \in \P$, let $\T(x, y, \b, \a)$ denote the torus
\begin{displaymath}
\T(x, y, \b, \a) := \{ (u_n, v_n)_{1 \leq n \leq N-1} | u_n^2 + v_n^2 = x_n^2 + y_n^2 \;\; \forall n \} \times \{ \b \} \times \{ \a \}.
\end{displaymath}

\begin{prop} \label{respfibr}
For any $(b,a) \in \M$,
\begin{equation} \label{isobatorba}
\Omega(\textrm{Iso} \, (b,a)) = \T(\Omega(b,a)).
\end{equation}
\end{prop}

\begin{proof}
Note that the action variables $I_n$ are defined in terms of the discriminant $\Delta_\l$, and $\Delta_\l$ is a spectral invariant. Hence $\Omega(\textrm{Iso}(b,a)) \subset \T(\Omega(b,a))$. As $\textrm{Iso}(b,a)$ and $\T(\Omega(b,a))$ are both tori of the same dimension, (\ref{isobatorba}) then follows.
\end{proof}

\begin{cor} \label{hamfunctact}
The pullback $\hat{H} = H \circ \Omega^{-1}$ of the Hamiltonian of the
periodic Toda lattice $H$ by $\Omega^{-1}$ is a function of the action variables $(I_n)_{1 \leq n \leq N-1}$ and the Casimir functions $C_1, C_2$ alone. In other words, $\hat{H}$ is in Birkhoff normal form. Therefore, the coordinates $((x_n, y_n)_{1 \leq n \leq N-1}, C_1, C_2)$ are global Birkhoff coordinates of the periodic Toda lattice.
\end{cor}

\begin{proof} Note that $H$ can be written as
\begin{displaymath}
  H = \frac{1}{2} \sum_{n=1}^N b_n^2 + \sum_{n=1}^N a_n^2 = \frac{1}{2} \, \textrm{tr} \, (L(b,a)^2) = \frac{1}{2} \sum_{j=1}^N (\l_j^+)^2
\end{displaymath}
where $(\l_j^+)_{1 \leq j \leq N}$ are the $N$ eigenvalues of $L(b,a)$. Hence, by Proposition \ref{respfibr}, $\hat{H} = H \circ \Omega^{-1}$ is constant on $\T(x, y, \b, \a)$, i.e. $\hat{H}$ is a function of $(I_n)_{1 \leq n \leq N-1}$, $C_1$, $C_2$ alone.
\end{proof}

\appendix

\section{Proof of Lemma \ref{est2}} \label{proofest}

We begin by proving the estimate (\ref{estb}). In a first step we use a trace formula observed by van Moerbeke \cite{moer} which expresses $b_1$ as a linear combination of the traces of $L^+ \equiv L^+(b,a)$, $L^- \equiv L^-(b,a)$ and $L_2 \equiv L_2(b,a)$, the matrix obtained from $L(b,a)$ by removing the first column and the first row:
\begin{equation} \label{trace}
b_1 = \frac{1}{2}(\textrm{tr }L^+ + \textrm{ tr }L^-) - \textrm{ tr }L_2.
\end{equation}
The Dirichlet eigenvalues $(\mu_n)_{1 \leq n \leq N-1}$ can be characterized as the spectrum of $L_2$. Hence we can rewrite (\ref{trace}) as
\begin{equation} \label{trace2}
b_1 = \frac{1}{2}(\l_1 + \l_{2N}) + \frac{1}{2}\sum_{k=1}^{N-1}(\l_{2k} + \l_{2k+1} - 2\mu_k).
\end{equation}
For $b_{n+1}$ with $1 \leq n \leq N-1$, a similar formula can be derived,
\begin{equation} \label{bnformula}
  b_{n+1} = \frac{1}{2}(\l_1 + \l_{2N}) + \frac{1}{2}\sum_{k=1}^{N-1}(\l_{2k} + \l_{2k+1} - 2\mu_k^{(n)}).
\end{equation}
%where we have replaced the Dirichlet spectrum $\mu_i$ by another auxiliary spectrum $\mu_i^{(n)}$, introduced in \cite{toda}, defined by $y_1(N+1,\mu_i^{(n)}|n) = 0$. Here, for any $n \in \Z$ fixed, $y_1(k,\l |n)$ and $y_2(k,\l| n)$ ($k \in \Z$) are the fundamental solutions of
%\begin{equation} \label{shift}
%(R_{S^n b,S^n a} y) (k, \l) = \l y(k, \l) \quad (k \in \Z)
%\end{equation}
%and $R_{S^n b, S^n a}$ is the difference operator associated with $L(S^n b, S^n a)$. Whereas the eigenvalues $\l_i^{(n)}$ of $Q(S^n b, S^n a)$ coincide with the ones for $Q(b,a)$, since the discriminants corresponding to $(S^n b, S^n a)$ and $(b,a)$ are the same, the eigenvalues $(\mu_k^{(n)})_{1 \leq k \leq N-1}$ do typically not coincide with $(\mu_k)_{1 \leq k \leq N-1}$.
Here $(\mu_i^{(n)})_{1 \leq i \leq N-1}$ denotes the Dirichlet spectrum of $L(S^n b, S^n a)$. Note that $(S^n b)_1 = b_{n+1}$ and
\begin{displaymath}
  \textrm{spec} \; Q(S^n b, S^n a) = \textrm{spec} \; Q(b,a),
\end{displaymath}
since the discriminants (\ref{discrdef}) for $(b,a)$ and $(S^n b, S^n a)$ coincide. Hence (\ref{bnformula}) can be obtained by applying (\ref{trace2}) to the element $(S^n b, S^n a)$ instead of $(b,a)$. Note that for any $1 \leq k \leq N-1$, the eigenvalue $\mu_k^{(n)}$ lies in the closed interval $[\l_{2k}, \l_{2k+1}]$.

It follows from (\ref{bnformula}) that the difference $b_i - b_j$ can then be estimated by
\begin{equation} \label{bibjdiff}
|b_i - b_j| \leq \sum_{k=1}^{N-1} \gamma_k(b,a).
\end{equation}
As for any $1 \leq n \leq N$ we have
\begin{displaymath}
  N b_n - N C_1 = N b_n - \sum_{j=1}^N b_j = \sum_{j=1}^N (b_n - b_j),
\end{displaymath}
it follows that $N |b_n| \leq N |C_1| + N \sum_{k=1}^{N-1} \gamma_k$, and formula (\ref{estb}) of Lemma \ref{est2} is established.

To obtain (\ref{esta}), choose an arbitrary $L^2$-orthonormal basis $(f_j)_{1 \leq j \leq N} \subseteq \R^N$ of eigenvectors associated to the eigenvalues $(\l_j^+)_{1 \leq j \leq N}$ of $L^+(b,a)$. We claim that
\begin{equation} \label{invab}
a_k = \sum_{j=1}^{N} \l_j^+ f_j(k) f_j(k+1).
\end{equation}
To verify (\ref{invab}), multiply
$$ a_{k-1} f_j(k-1) + b_k f_j(k) + a_k f_j(k+1) = \l_j^+ f_j(k) $$
by $f_j(k+1)$ and sum over $j$. As $(f_j)_{1 \leq j \leq N}$ is an orthonormal basis, the $N \times N$-matrix $(f_j(k))_{j,k}$ is orthogonal, hence $\sum_{j=1}^N f_j(k) f_j(l) = \delta_{kl}$ for $1 \leq k,l \leq N$, and (\ref{invab}) follows. Since $|f_j(k)| \leq 1 \; \forall \, j,k$, the identity (\ref{invab}) implies that
\begin{equation} \label{esta2}
 a_k \leq \sum_{j=1}^{N} |\l_j^+|.
\end{equation}
To estimate $|\l_j^+|$, note that
\begin{displaymath}
  \l_j^+ - C_1 = \l_j^+ - \frac{1}{N} \sum_{k=1}^N \l_k^+ = \frac{1}{N} \sum_{k=1}^N (\l_j^+ - \l_k^+).
\end{displaymath}
Hence for any $1 \leq j \leq N$,
\begin{displaymath}
  -(\l_{2N} - \l_1) \leq \l_j^+ + C_1 \leq \l_{2N} - \l_1 \quad \textrm{or} \quad |\l_j^+| \leq |C_1| + (\l_{2N} - \l_1)
\end{displaymath}
leading to
\begin{eqnarray}
|\l_j^+| \; & \leq \; & |C_1| + \sum_{n=1}^{2N-1} (\l_{n+1} - \l_n) \nonumber\\
& = & |C_1| + \sum_{n=1}^{N-1} \gamma_n + \sum_{n=1}^{N} \big( \l_{2n} - \l_{2n-1} \big). \label{esta3}
\end{eqnarray}
%By constructing an auxiliary conformal map and using tools from
%complex analysis, in particular the maximum principle for harmonic
%measures, one can obtain a uniform upper bound for the terms in
%the second sum of the latter expression, which is
We now recall from (\cite{ahtk1}, Appendix B) that for any $1 \leq n \leq N$,
\begin{equation} \label{esta4}
\l_{2n} - \l_{2n-1} \leq \frac{2 \pi C_2}{N}.
\end{equation}
Combining (\ref{esta3}) and (\ref{esta4}) yields
\begin{displaymath}
|\l_j^+| \leq |C_1| + \sum_{n=1}^{N-1} \gamma_n + 2 \pi C_2.
\end{displaymath}
Substituting this inequality into (\ref{esta2}) leads to the desired estimate (\ref{esta}). Finally, the claimed estimate (\ref{estl}) easily follows from (\ref{esta4}). Hence Lemma \ref{est2} is proved.

\vspace{.4cm}

\textsc{Institut f\"ur Mathematik, Universit\"at Z\"urich, Winterthurerstrasse 190, CH-8057 Z\"urich, Switzerland} \\
\emph{E-mail address:} \texttt{andreas.henrici@math.unizh.ch}

\vspace{.4cm}

\textsc{Institut f\"ur Mathematik, Universit\"at Z\"urich, Winterthurerstrasse 190, CH-8057 Z\"urich, Switzerland} \\
\emph{E-mail address:} \texttt{thomas.kappeler@math.unizh.ch}


\begin{thebibliography}{99}

%\input{biblio.tex}

\bibitem{bbgk}
\textsc{D. B\"attig, A. M. Bloch, J. C. Guillot \& T. Kappeler}, On the
symplectic structure of the phase space for periodic KdV, Toda,
and defocusing NLS. \emph{Duke Math. J.} \textbf{79} (1995),
549-604.

\bibitem{bggk}
\textsc{D. B\"attig, B. Gr\'ebert, J. C. Guillot \& T. Kappeler},
Fibration of the phase space of the periodic Toda lattice.
\emph{J. Math. Pures Appl.} \textbf{72} (1993), 553-565.

\bibitem{fla1}
\textsc{H. Flaschka}, The Toda lattice. I. Existence of integrals.
\emph{Phys. Rev.}, Sect. B \textbf{9} (1974), 1924-1925.

\bibitem{flmc}
\textsc{H. Flaschka \& D. McLaughlin}, Canonically conjugate
variables for the Korteweg-de Vries equation and the Toda lattice
with periodic boundary conditions. \emph{Prog. Theor. Phys.}
\textbf{55} (1976), 438-456.

\bibitem{fpu}
\textsc{E. Fermi, J. Pasta \& S. Ulam}, Studies of non linear problems. \emph{Los Alamos Rpt.} \textbf{LA-1940} (1955). In: \emph{Collected Papers of Enrico Fermi}. University of Chicago Press, Chicago, 1965, Volume II, 978-988. Theory, Methods and Applications, 2nd ed., Marcel Dekker, New York, 2000.

\bibitem{gatr1}
\textsc{J. Garnett \& E. Trubowitz}, Gaps and bands of one
dimensional periodic Schr\"odinger operators. \emph{Comm. Math.
Helv.} \textbf{59} (1984), 258-312.

\bibitem{gkp}
\textsc{B. Gr\'ebert, T. Kappeler \& J. P\"oschel}, Normal form theory for the NLS equation: a preliminary report. Preprint, 2003.

\bibitem{ahtk1}
\textsc{A. Henrici \& T. Kappeler}, Global action-angle variables for the
periodic Toda lattice. \texttt{arXiv:0802.4032v1 [nlin.SI]}.

\bibitem{ahtk3}
\textsc{A. Henrici \& T. Kappeler}, Birkhoff normal form for the periodic Toda lattice. \texttt{arXiv:nlin/0609045v1 [nlin.SI]}. To appear in \emph{Contemp. Math.}

\bibitem{kama}
\textsc{T. Kappeler \& M. Makarov}, On Birkhoff coordinates for KdV.
\emph{Ann. Henri Poincar\'{e}} \textbf{2} (2001), 807-856.

\bibitem{kapo}
\textsc{T. Kappeler \& J. P\"oschel}, \emph{KdV~\&~KAM}. Ergebnisse der Mathematik, 3. Folge, vol. \textbf{45}, Springer, 2003.

\bibitem{kato}
\textsc{T. Kappeler \& P. Topalov}, Global Well-Posedness of KdV in $H^{-1}(\T, \R)$. \emph{Duke Math. J.} \textbf{135}(2) (2006), 327-360.

\bibitem{mana}
\textsc{S. V. Manakov}, Complete integrability and stochastization
of discrete dynamical systems. \emph{Zh. Exp. Teor. Fiz.}
\textbf{67} (1974), 543-555 [Russian]. English translation:
\emph{Sov. Phys. JETP} \textbf{40} (1975), 269-274.

\bibitem{mcva1}
\textsc{H. P. McKean \& K. L. Vaninsky}, Action-angle variables for
the cubic Schroedinger equation. \emph{Comm. Pure Appl. Math.}
\textbf{50} (1997), 489-562.

\bibitem{moer}
\textsc{P. van Moerbeke}, The spectrum of Jacobi matrices.
\emph{Invent. Math.} \textbf{37} (1976), 45-81.

\bibitem{teschl2}
\textsc{G. Teschl}, \emph{Jacobi Operators and Completely
Integrable Nonlinear Lattices}. Math. Surveys and Monographs
\textbf{72}, Amer. Math. Soc., Providence, 2000.

\bibitem{toda}
\textsc{M. Toda}, \emph{Theory of Nonlinear Lattices}, 2nd enl.
ed., Springer Series in Solid-State Sciences \textbf{20},
Springer, Berlin, 1989.


\end{thebibliography}
\end{document}